\documentclass[11pt]{article}
\usepackage{amsmath, amsfonts, amssymb, amsbsy, amsthm, enumerate, verbatim}
\usepackage[mathscr]{eucal}
\usepackage[english]{babel}

\usepackage{graphicx}
\newtheorem{theo}{Theorem}[section]
\newtheorem{prop}[theo]{Proposition}
\newtheorem{lemm}[theo]{Lemma}

\newtheorem{defi}[theo]{Definition}
\theoremstyle{definition}
\newtheorem{rema}[theo]{Remark}

\newcommand{\bel}{\begin{equation} \label}
\newcommand{\ee}{\end{equation}}

\def\beq{\begin{equation}}
\def\eeq{\end{equation}}
\newcommand{\bea}{\begin{eqnarray}}
\newcommand{\eea}{\end{eqnarray}}
\newcommand{\beas}{\begin{eqnarray*}}
\newcommand{\eeas}{\end{eqnarray*}}

{

\begin{document}

\begin{center}
{\Large \bf Spectral monodromy of non selfadjoint operators}

\medskip

\today
\end{center}

\medskip

\begin{center}
{\sc Quang Sang Phan}\\
\end{center}


%



\begin{abstract}  ~\\
     We propose to build in this paper a combinatorial invariant, called the "spectral monodromy" from the spectrum of a single (non-selfadjoint) $ h $-pseudodifferential operator with two degrees of freedom in the semi-classical limit. \\
     Our inspiration comes from the quantum monodromy defined for the joint spectrum of an integrable system of $ n $ commuting selfadjoint $ h $-pseudodifferential operators, given by S. Vu Ngoc. \\
     The first simple case that we treat in this work is a normal operator. In this case, the discrete spectrum can be identified with the joint spectrum of an integrable quantum system. \\
     The second more complex case we propose is a small perturbation of a selfadjoint operator with a classical integrability property.
     We show that the discrete spectrum (in a small band around the real axis) also has a combinatorial monodromy.
     The difficulty here is that we do not know the description of the spectrum everywhere, but only in a Cantor type set. In addition, we also show that
     the monodromy can be identified with the classical monodromy (which is defined by J. Duistermaat). These are the main results of this article.

 ~\\
     {\bf  Keywords:} Non-selfadjoint, integrable system, spectral analysis, pseudo-differential operators, Birkhoff
normal form, asymptotic spectral
\end{abstract}


\section{Introduction}

\subsection{General framework}
            This paper aims at understanding the structure of the spectrum of some classes of non-selfadjoint operators in the semi classical limit.
             It is a quantum problem that we treat with the help of semi-classical techniques
             combined with the general spectral theory of pseudo-differential operators.
            We will also make the link with classical results that illuminate the initial quantum problem.

\subsection{Monodromy of $ h $-pseudo-differential non-selfadjoint operators}

More explicitly, in this paper, we propose to build a new  characteristic objet of the structure of the spectrum of non-selfadjoint $ h $-pseudo-differential operators in the semi-classical limit.

Our inspiration comes from quantum monodromy, which is defined for the joint spectrum (see the definition
\ref{jspec}) of a system of $ n $ $h-$ pseudo-differential operators that commute (i.e a completely integrable
quantum system). This is a quantum invariant given by San Vu Ngoc  \cite{Vu-Ngoc99}(or
\cite{Vu-Ngoc.S01}).  \\
Under certain conditions, the joint spectrum on a domain $ U $ of regular values of the map of principal
symbols is an asymptotic affine lattice in the sense that there is an invertible symbol, denoted $ f_ {\alpha} $,
from any small ball $ B_ {\alpha} \subset U $ in $ \mathbb R ^ n $ that sends the joint spectrum to $ \mathbb
Z ^ n $ modulo $\mathcal O(h^\infty)$ (a result of Charbonnel \cite{Charbonnel88}). These $ (f_ {\alpha}, B_
{\alpha}) $ are considered as local charts of $ U $ whose transition functions, denoted by $ A_ {\alpha  \beta}
$, are in the integer affine group $ GA (n, \mathbb Z) $. The quantum monodromy is defined as the $
1$-cocycle $ \{A_ {\alpha \beta} \} $ modulo-coboundary in the \v{C}ech cohomology
               $\check{H}^1(U,GA(n, \mathbb Z) )$. \\
For details of this monodromy, we can see the article \cite{Vu-Ngoc99} or section \ref{qu mono} of this paper.

            Since this work, a mysterious question remains open:
              Can we define (and detect) such an invariant for a single $h$-pseudo-differential operator?
              If this happens, we will call it the "spectral monodromy".

        From a geometrical point of view, since the joint spectrum of a quantum integrable system
               is an asymptotic affine lattice,
        if one realizes the parallel transport on the lattice of a basic rectangle with a vertex $ c $ along a some closed path $\gamma_c$ (of base $c$) and returns to the starting point, then the initial rectangle can become a different rectangle (see figure below).
    It is the existence of quantum monodromy. Contrariwise, by the spectrum of a selfadjoint operator being contained in a straight line (real axis), it seems impossible to define such a parallel transport for a single operator. It is not known how to define a monodromy in this case.  \ \

\begin{figure}
\begin{center}
\includegraphics [height=90mm,angle=0]{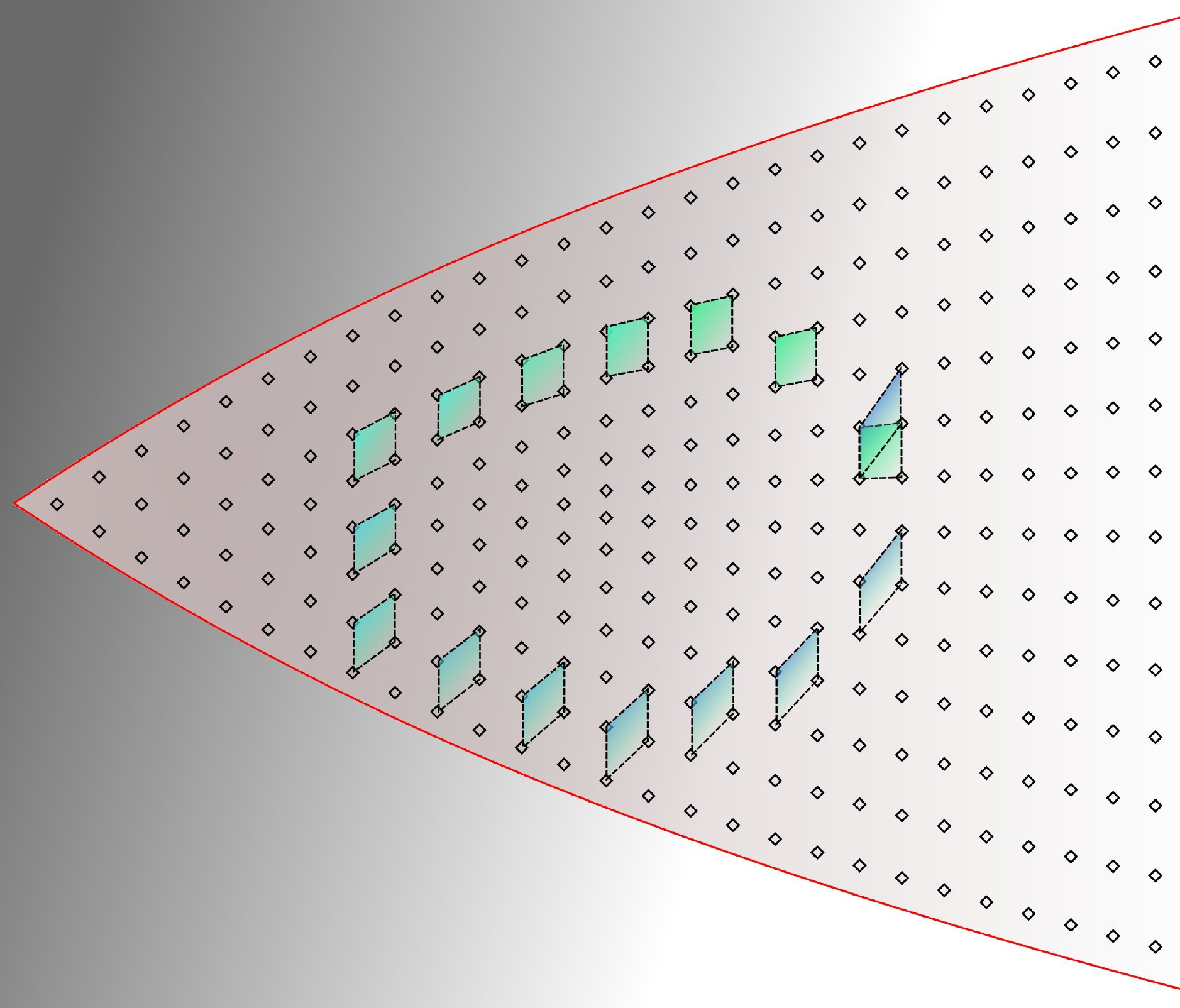}
\caption{Joint spectrum of the spherical pendulum with monodromy (Image by S. Vu Ngoc) }
\end{center}
\end{figure}

 ~~\\
However, if we perturb a self-adjoint operator by a non-symmetric term, the spectrum becomes complex,
and we may hope to find a geometric structure (lattice, monodromy ...). We propose in this paper to apply
this idea to certain classes of  $ h $-pseudo-differential operators of two
              degrees of freedom.

\subsubsection{The first case}

  The first simple case that we propose is a $ h$-pseudo-differential operator of form $ P_1 + i P_2 $  with two self-adjoint operators $ P_1, P_2 $ that commute.
 This is the form of a normal operator.
 We show in section 3 that the discrete spectrum of $ P_1 + i P_2 $ is identified with the joint spectrum of the integrable quantum system $ (P_1, P_2) $
 (see Theorems \ref{bac cau} and \ref{bac cau 1}).
One can simply define the "affine spectral monodromy" of operator $ P_1 + i P_2$  as the quantum monodromy of the joint spectrum. For details, see section \ref{mono op normal}.
\subsubsection{The second case}
The second case that we study is more complex.
For the quantum monodromy and thus the affine spectral monodromy (in the first case), the quantum integrability condition of $ P_1, P_2 $ is necessary but seems a bit heavy
to have a monodromy
because the quantum monodromy has a relationship with the classical monodromy (which is given by J.Duistermaat \cite{Duis80}) of the integrable classical system $ (p_1, p_2) $, the corresponding principal symbols of $ P_1, P_2 $. It is a result in article \cite{Vu-Ngoc99}.

For this reason, keeping the property of classical integrability, we will propose to consider a small perturbation of a self-adjoint operator of the form $P_\varepsilon:= P_1+ i \varepsilon P_2$ assuming that the principal symbols
 $p_1, p_2 $ commute, $\varepsilon \rightarrow 0$ and in the regime
 $h \ll \varepsilon = \mathcal{O}(h^\delta)$ for $0< \delta <1 $.

Here there is no joint spectrum, so we can not apply the construction of the quantum monodromy. However
with the help of the results of asymptotic spectral theory of M.Hitrik, J.Sjöstrand, S. Vu Ngoc (\cite{Hitrik04},
\cite{Hitrik05}, \cite{Hitrik08} and especially \cite{Hitrik07}) (under conditions detailed in the section 4) by
revisiting the procedure of Birkhoff normal form (section (\ref{sec FNB})), the spectrum of $P_\varepsilon$ is
located in a horizontal band of height $\mathcal{O}(\varepsilon)$ and in this band, we can give the
asymptotic expansion of eigenvalues of $P_\varepsilon$ in some "good rectangles"
$R(\chi_a,\varepsilon,h)$ (see definition \ref{dn bonnes valeurs}) of size $\mathcal{O} ^\delta \times
\mathcal{O }(\varepsilon h^\delta)$ which are associated with Diophantine torus $\Lambda_a$.

There is a correspondence between
  $\lambda \in \sigma(P_\varepsilon) \cap R(\chi_a,\varepsilon,h) $ and $hk$ in a part of
 $h \mathbb Z^2 $ by a diffeomorphism (a micro-chart) of form (see the formula (\ref{new hk})):
\begin{eqnarray}
     f :  R(\chi_a,\varepsilon,h) & \rightarrow & E(a,\varepsilon,h)
  \nonumber \\
    \sigma(P_\varepsilon) \cap R(\chi_a,\varepsilon,h) \ni  \lambda & \mapsto &
  f(\lambda,\varepsilon; h) \in h \mathbb Z^2 +\mathcal O(h^\infty).
\end{eqnarray}
For more details of this idea, see the section \ref{sec op comm}.

Nevertheless all Diophantine tori do not quite fill the phase space (see \cite{Broer10}, \cite{Poschel82}) and despite its
density, is not yet known whether such a expansion holds globally on any small domain of spectral band.
However, we will prove the global existence (for such rectangles) of
the first term of this expansion on any small area in the set of
regular values of the application $(p,\varepsilon q)$.

The spectrum of $P_\varepsilon$ is the model of a particular set
  $\Sigma(\varepsilon, h)$ on a domain $ U(\varepsilon)$ that we define in the section 4 and call "asymptotic pseudo-lattice" ( see definition \ref{pseu-cart})
whose differential transition functions between the adjacent "pseudo-locals charts" are  in the group
$GL(2, \mathbb Z)$ modulo $\mathcal O(\varepsilon, \frac{h}{\varepsilon})$.

This allows us to treat the inverse problem: define a combinatorial invariant (the spectral monodromy) from
the spectrum of $P_\varepsilon$. This is the main result of this paper, presented in section 3.

It would be very interesting to extend these results to the case where $ p $ is a perturbation an integrable system, using the work of Broer, Cusham, Fass\`{o} et Takens \cite{Broer07}.

\section{Affine Spectral Monodromy}
            ~~\\
                The quantum monodromy that is defined for the discrete joint spectrum of a integrable quantum system
                of $ n $ commuting $h-$pseudo-differential operators is completely given by S. Vu Ngoc \cite{Vu-Ngoc99}.

                We propose to define the monodromy for a single $h-$pseudo-differential operator and in this section, we will treat the simple case of a normal operator. To do this, we will give an identification between
                the discrete spectrum of a normal operator and the joint spectrum of an integrable quantum system (theorems \ref{bac cau} and \ref{bac cau 1}).

                First we briefly recall the standard class of  $h$-(Weyl-)pseudo-differential operators which is used through in this article.
                Then we give the results of spectral theory that allow us to define the "affine spectral monodromy" of a pseudo-differential normal operator.

\subsection{Pseudo-differential operators}

                We will work throughout this article with pseudo-differentials operators obtained by the $h-$Weyl-quantization with standard classes of symbols on $M= T^* \mathbb R^n=\mathbb R^{2n}_{(x,\xi)}$.
                These operators admit the standard properties of pseudo differential operators.
                For more details, see the references \cite{Dimas99}, \cite{Robert87}, \cite{Shubin01}.

        \begin{defi} \label{fonc ord}
                        A function $m: \mathbb R^{2n} \rightarrow (0, + \infty)$ is called an order function
                        (or tempered weight in the book of D. Robert \cite{Robert87}) if there are constants  $C,N >0$ such that
                                $$m(X)  \leq C \langle X-Y\rangle^{ N} m(Y), \forall X,Y \in \mathbb R^{2n},$$
        with notation $\langle Z\rangle= (1+ |Z|^2)^{1/2}$ for $Z \in \mathbb R^{2n}$.
        \end{defi}

        One use often the order function  $m(Z) \equiv 1$ or $$m(Z)= \langle Z \rangle ^{l/2}= (1 + |Z|^2 )^{l/2},$$ with a given constant $l \in \mathbb R $.

         \begin{defi}
                        Let $m$ be an order function and $k \in \mathbb R$, we define classes of symbols of $h$-order $k$, $S^k(m)$ (families of functions) of $(a(\cdot;h))_{h \in (0,1]}$ on $\mathbb R^{2n}_{(x,\xi)}$ by
                        \begin{equation}
                                S^k(m)= \{ a \in C^\infty (\mathbb R^{2n})
                                 \mid  \forall \alpha \in \mathbb N ^{2n}, \quad |\partial^\alpha a | \leq  C_\alpha h^k m \} ,
                        \end{equation}
         for some constant $C _\alpha >0$, uniformly in $h \in (0,1]$. \\
         A symbol $a$ is called $\mathcal O(h^\infty)$ if it's in $\cap _{k \in \mathbb R } S^k(m):= S^{\infty}(m) $.
        \end{defi}

        Then $ \Psi^k(m)(M)$ denotes the set of all (in general unbounded) linear operators $A_h$ on $L^2(\mathbb R^n)$, obtained from the $h-$Weyl-quantization of symbols $a(\cdot;h) \in S^k(m) $ by the integral:
        \begin{equation} \label{symbole de W}
                            (A_h u)(x)=(Op^w_h (a) u)(x)= \frac{1}{(2 \pi h)^n}
                                 \int_{ \mathbb R^{2n}} e^{\frac{i}{h}(x-y)\xi}
                                 a(\frac{x+y}{2},\xi;h) u(y) dy d\xi.
        \end{equation}

        In this paper, we always assume that the symbols admit a classical asymptotic expansion in integer powers of $h$. The leading term in this expansion is called the principal symbol of the operator.

\subsection{Quantum monodromy of Integrable quantum systems}  \label{qu mono}

        If an integrable quantum system $(P_1(h), \ldots, P_n(h))$ with joint principal symbol $p$ is proper, then near a
        regular value of $ p $, the joint spectrum of the system locally has the structure of an affine integral lattice \cite{Charbonnel88}, \cite{Colin80}. By S. Vu Ngoc, this leads to the construction of quantum monodromy- a natural invariant defined from the spectrum, see \cite{Vu-Ngoc99}.
        The non-triviality of this invariant obstructs the global existence of lattice structure of the joint spectrum.
        More explicit:

        Let $X$ a compact manifold of dimension $ n $ or $X= \mathbb R^n$ and let $M:=T^*X$ the tangent bundle of $X$. Let an integrable quantum system $(P_1(h), \ldots, P_n(h))$ of $n$ commuting selfadjoint $h-$Weyl pseudo-differential operators on $L^2(X)$:  $[P_i(h),P_j(h)]=0$.
        We will assume that these $P_j(h)$ are in $\Psi^{0}(M)$, classical and of order zero. In any coordinate chart their Weyl symbols $p_j(h) $ have an asymptotic expansion of the form:
             $$p_j(x,\xi;h)= p_0^j(x,\xi)+ h p_1^j(x,\xi) + h^2 p_2^j(x,\xi)+ \cdots. $$
        Assume that the differentials of the principal symbols $p_0^j$ are linearly independent almost everywhere on $M$.
        The map of joint principal symbols $p=(p_0^1,\ldots, p_0^n)$ is a momentum map with respect to the symplectic Poisson bracket on $ T ^ * X $ ($\{p_0^i,p_0^j \}= 0$).
        We will assume that $p$ is proper.

        Let $U_r$ be an open subset of regular values of $ p $ and let $ U $ be a
                  certain open subset with compact closure $K: = \overline{U} $ in $U_r$.
        We define the joint spectrum of the system in $K$, noted $\sigma_{conj}(P_1, \ldots, P_n)$ by:
                \begin{equation}  \label{jspec}
                    \sigma_{conj}(P_1(h), \ldots, P_n(h)) = \{ (E_1(h), \ldots, E_n(h)) \in
                    K | \cap_{j=1}^n Ker(P_j(h) - E_j(h)) \neq \emptyset \}.
                \end{equation}
        Let $ \Sigma(h)= \sigma_{conj}(P_1(h), \ldots,P_n(h)) \cap U $.
        Is is known from the work of Colin de Verdi\`{e}re \cite{Colin80} and Charbonnel \cite{Charbonnel88}, $ \Sigma(h)$ is discrete and for small $h$ is composed of simple eigenvalues.
        Moreover, $ \Sigma(h)$ is "an asymptotic affine lattice" on $U$ in the sense:  there are locally invertible symbols of order zero, denoted $ f_ {\alpha}(\cdot; h) $
        from any small ball $ B_{\alpha} \subset U $ in $ \mathbb R^ n$, sending $ \Sigma (h) $ in $h \mathbb Z^n$ modulo $\mathcal O(h^\infty)$.
        These $ (f_ {\alpha} ,B_ {\alpha}) $ are considered as locals charts of $ \Sigma (h) $ on $ U $ whose transition functions, denoted by $ A_ {\alpha  \beta} $
        are in the integer affine group $ GA (n, \mathbb Z) $.
        The quantum monodromy is defined as the $ 1$-cocycle $ \{A_ {\alpha \beta} \} $ modulo-coboundary in the \v{C}ech cohomology
        $\check{H}^1(U,GA(n, \mathbb Z) )$ (see following definition ). We denote $$[\mathcal M _{qu}] \in \check{H}^1(U,GA(n, \mathbb Z) ). $$


\begin{figure}[!h]
\begin{center}
 \includegraphics[width=0.8\textwidth]{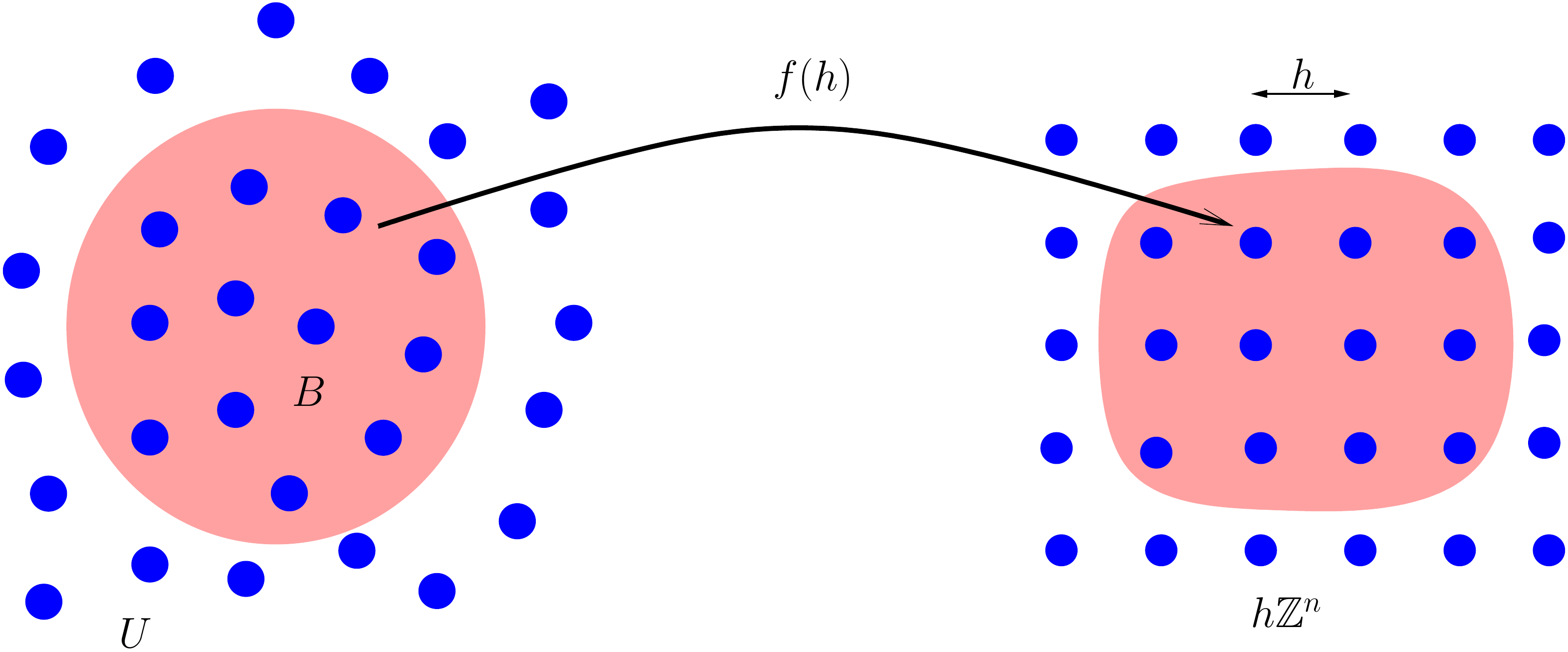}\\
\caption{Asymptotic affine lattice }
\end{center}
\end{figure}

        \begin{rema}
        We recall here the definition of \v{C}ech cohomology that is use often in this paper. Let  $M$ be a manifold and $(G,.)$ be a group. Assume that $ \{ U_\alpha \}_ { \alpha \in I}  $ is a locally finite cover of open sets of $M$ such that the intersection of a finite number of $ U_ \alpha $ is either contractible or empty.
        We denote $\mathcal{C}^0 (M,G)$ the set of $0$-cochains and $\mathcal{C}^1 (M,G)$ the set of $1$-cochains with values in $G$ by:
                 $$\mathcal{C}^0 (M,G)= \{ (c_\alpha)_ { \alpha \in I} \mid c_\alpha \in G \},$$
                 $$\mathcal{C}^1 (M,G)= \{ (c_{\alpha, \beta }) \mid U_\alpha \cap U_\beta \neq \emptyset , c_{\alpha, \beta } \in G \}.$$
        We denote $\check{Z}^1(M, G)$ the set of $1$-cochains satisfying the cocycle condition
           $$c_{\alpha, \beta }. c_{\beta, \gamma }= c_{\alpha, \gamma } $$ if
            $ U_\alpha  \cap  U_\beta \cap U_\gamma \neq \emptyset $.
        ~~\\
        We define an equivalence relation, denoted $"\sim"$ on $\check{Z}^1(M, G)$: two cocycles
        $(c_{\alpha, \beta })$ and $(c'_{\alpha, \beta })$ are equivalent $(c_{\alpha, \beta }) \sim
        (c'_{\alpha, \beta })$ if there exists a $0$-cochain $(d_\alpha) \in \mathcal{C}^0 (M,G)$ such that
        $c'_{\alpha, \beta }= d_\alpha . c_{\alpha, \beta }. d_\beta ^{-1} $ for any $U_\alpha \cap
        U_\beta \neq \emptyset $.
        ~~\\
        Then the \v{C}ech cohomology of $M$ with values in $G$ is the quotient set
        $$\check{H}^1(M, G)= \check{Z}^1(M, G) / \sim .$$
        ~~\\
        Note that it does not depends on choice of cover $ \{ U_\alpha \}_ { \alpha \in I}$.
        \end{rema}

         ~~\\
        The quantum monodromy can be considered as a group morphism (holonomy):
               \begin{equation} \mu: \pi_1(U) \rightarrow GA(n, \mathbb Z) / \{ \sim \}
               \end{equation}
        which is the product of transition functions along a closed loop modulo by conjugation $"\sim "$.

        For more details and discussion of this monodromy, we can see \cite{Vu-Ngoc99}, \cite{Vu-Ngoc.S01}.

\subsection{Normal operators}

             In this section,  we will show the natural statement that: the discrete spectrum of a unbounded normal operator $A$ can be identified with the joint spectrum of an integrable system which consists of the real part and the imaginary part of $A$. This allows us to define the monodromy of normal operator as an application of quantum monodromy.

             Consider a normal operator (usually unbounded) $ A $ with dense domain $D (A) = D $ on a Hilbert space $H$. It is known that the adjoint operator $ A ^ * $ has the same domain $ D (A ^ *) = D (A) = D $.
             We denote the real part and imaginary part of $ A $ by $ A = A_1 +  iA_2$ with
                     \begin{equation} \label{ao-thuc}
                        A_1= \frac{A+A^*}{2}, A_2= \frac{A-A^*}{2i}, D(A_1)=D(A_2)= D.
                    \end{equation}
                    ~~\\
             It is true that $ A_1 $ and $ A_2 $ defined by the formula (\ref{ao-thuc}) are self-adjoint. Moreover, the commutativity of $ A $ and $ A ^ * $ is equivalent to the commutativity of $ A_1, A_2 $ and therefore $A_1A_2=A_2A_1$.

             In this article, one say that two selfadjoint operators (usually unbounded) $(A_1,D(A_1))$ and
$(A_2,D(A_2))$ commute if $A_1 A_2=A_2A_1$ and this definition includes the requirement of domain:
                    $$Ran(A_2) \subseteq D(A_1), Ran(A_1) \subseteq D(A_2), D(A_1 A_2)= D(A_2A_1).$$
             Conversely, let two commuting self-adjoint operators $ A_1 $ and $A_2 $, $ D(A_1) = D (A_2) = D $ dense in $ H $. Then the operator defined by $A:=A_1+iA_2$, $D(A)=D$ is a well-defined normal operator (and hence closed) on $H$.

            In the literature, there are alternative definitions of discrete spectrum (see\cite{Davies95}, \cite{Dimas99}, \cite{Oliveira09}, \cite{Havl08}) which coincide in the self-adjoint case.
             In this article, we use the following general definition of discrete spectrum, see \cite{RS78}, \cite{Kato95}.

             \begin{defi}
                 For a closed operator $(A, D(A))$, let $ \lambda $ be an isolated point of $ \sigma (A) $: there is
                 $ \varepsilon > 0$ such that
                 $ \{\mu \in \mathbb C,\mid | z- \lambda| < \varepsilon \} \cap \sigma(A)= \{\lambda\}$.
                 For all $ 0 <r <\varepsilon $,  we can define the projection $P_ \lambda $ (not necessarily orthogonal) by
                     \begin{equation}
                        P_\lambda= \frac{1}{2 \pi i } \int_{| z- \lambda|= r} (z-A)^{-1}.
                     \end{equation}
                 We say that $\lambda \in \sigma(A)$ is in the discrete spectrum of $A$, denoted $\sigma_{disc}(A)$ if $\lambda$ is isolated in $\sigma(A)$ and $ P_\lambda$ has finite
                 rank. \\
                 We define the essential spectrum of $A$, denoted $\sigma_{ess}(A)$ as the complement of the discrete spectrum $$\sigma_{ess}(A)= \sigma(A) \setminus \sigma_{disc}(A). $$
            \end{defi}

            \begin{prop} \label{sp ess}
                Let $A_1,A_2$ two commuting self-adjoint operators on a Hilbert space $H$ with dense domain $D(A)=D(B)=D$, $A_1A_2= A_2A_1$. Then we have:
                     \begin{enumerate}
                        \item  If $ \lambda \in \sigma(A_1+iA_2)$, then $Re(\lambda) \in \sigma(A_1)$ and
                             $ Im(\lambda) \in \sigma(A_2)$.
                        \item If $ \lambda \in \sigma_p(A_1+iA_2)$, then $Re(\lambda) \in
                                            \sigma_p(A_1)$ and $ Im(\lambda) \in \sigma_p(A_2)$.   \\
                        Moreover if $\lambda$ is an eigenvalue of infinite multiplicity of $A_1+iA_2$, then $Re(\lambda),Im(\lambda) $ are the eigenvalues of infinite multiplicity corresponding of $A_1,A_2$.
                        \item If $ \lambda \in \sigma_{ess}(A_1+iA_2)$, then $Re(\lambda) \in \sigma_{ess}(A_1)$ and $ Im(\lambda) \in \sigma_{ess}(A_2)$.
                     \end{enumerate}
            \end{prop}

            \begin{proof}[Proof]
                        Let $A=A_1+i A_2$, $D(A)= D$.
                        As we explained earlier at the beginning of this section, $A$ is a normal operator with adjoint operator $A^*= A_1-iA_2$ .
                        For any complex number $\lambda \in \mathbb C$, $A-\lambda$ is still a normal operator. Then for any $u \in D$, the equality $\|(A-\lambda)u\|= \|(A^*- \overline{\lambda})u\|$ gives us
                            \begin{equation} \label{pt1}
                                \|(A_1+iA_2-\lambda)u\|^2=\|(A_1-Re(\lambda))u\|^2 +  \|(A_2-Im(\lambda))u\|^2
                            \end{equation}
                        This equation implies that:
                                \begin{equation} \label{pt2}
                                        Ker(A-\lambda) = Ker( A_1-Re(\lambda))\cap Ker(A_2-Im(\lambda)).
                                \end{equation}
                        If $\lambda \in \sigma (A_1+iA_2) $, by the Weyl theorem (see \cite{Havl08}, page $102$)), there exists a Weyl sequence for $A$ and $\lambda$: a sequence $u_n \in D$, $\|u_n\|=1$ such that $ \lim_{n \rightarrow \infty} \| (A-\lambda)u_n\|=0$.
                         By the equation (\ref{pt1}), it is still a Weyl sequence for $ A_1$ and $Re(\lambda)$, for  $A_2$ and $Im(\lambda)$. Again by the Weyl theorem, we have $Re(\lambda) \in \sigma(A_1)$ and $ Im(\lambda) \in \sigma(A_2)$.

                   If $ \lambda $ is an eigenvalue of $A_1+iA_2$, by the equation (\ref{pt2}) we have $Re(\lambda) \in \sigma_p(A_1)$ and $ Im(\lambda) \in \sigma_p(A_2)$. Moreover, it is obvious that if $\lambda$ is an eigenvalue of infinite multiplicity of $A_1+iA_2$, then $Re(\lambda),Im(\lambda) $ are also eigenvalues of infinite multiplicity corresponding to $ A_1, A_2 $.

                   We also note that if $ u $ is an eigenvector corresponding to $ \lambda $ of
                         $ A_1 + iA_2$, then $ u $ is also simultaneous eigenvector of $ A_1, A_2 $.

                   If $ \lambda \in \sigma_{ess}(A_1+iA_2)$, then there exists a orthogonal Weyl sequence for $A_1+iA_2$ in $\lambda$ such that:
                   $u_n \in D$, $\|u_n\|=1$ and $ \lim_{n \rightarrow \infty} \| (A-\lambda)u_n\|=0$.
                   By the equation (\ref{pt1}) and by the Weyl theorem for a self-adjoint operator (see \cite{Dimas99}, \cite{Oliveira09} page $287$), \cite{Havl08} page $173$, \cite{Davies95}... ), we obtain
                   $Re(\lambda) \in \sigma_{ess}(A_1)$ and $ Im(\lambda) \in \sigma_{ess}(A_2)$.
                   The proposition is shown.
            \end{proof}

            One can easily show that the reverse implications are false in general.
            ~~\\
             From this proposition, by identifying  $\mathbb C \cong \mathbb R^2$ we have the following result:

             \begin{theo} \label{bac cau 1}
                    Let $A_1,A_2$ two commuting self-adjoint operators on a Hilbert space $H$ with dense domain $D(A)=D(B)=D$. Let $I_1, I_2$ be two intervals of $\mathbb R$ such that the corresponding spectra
                    of $A_1$,$A_2$ in $I_1, I_2$ are discrete.

                    Then the spectrum of $A_1+iA_2$ in $  I_1 +i I_2 \cong  I_1 \times I_2 $ is discrete and
                         $$ \sigma(A_1+iA_2) \cap (I_1 +i I_2) \cong \sigma_{conj}(A_1,A_2) \cap (I_1 \times I_2). $$
             \end{theo}

             \begin{proof}
                        We always have inclusion:
                        $$\sigma_{conj}(A_1,A_2) \cap I_1 \times I_2 \subseteq \sigma_p(A_1+iA_2) \cap I_1 +i I_2
                         \subseteq \sigma(A_1+iA_2) \cap I_1 +i I_2.$$
                        Let us consider the inverse inclusion. For any $\lambda \in \sigma(A_1+iA_2) \cap I_1 +i
                        I_2$, the previous proposition says that: $Re(\lambda)\in \sigma(A_1+) \cap
                        I_1$ and $Im(\lambda) \in \sigma(A_2) \cap I_2$. \\
                        Because the corresponding spectra of $A_1,A_2$ in $I_1,I_2$ are
                        discrete, $Re(\lambda)$ and $Im(\lambda)$ are respectively isolated eigenvalues of
                         finite multiplicity of $A_1,A_2$. \\
                        From the equation (\ref{pt2}), $\lambda$ must be an eigenvalue of
                         finite multiplicity of $A_1+iA_2$ and there exists a common eigenvalue for
                        $A_1,A_2$:  $u \in D, \|u\|=1, A_1u=Re(\lambda)u, A_2u=Im(\lambda)u$.
                        Therefore
                        $$\lambda= (Re(\lambda),Im(\lambda) ) \in \sigma_{conj}(A_1,A_2) \cap I_1 \times I_2.$$
             \end{proof}

             We will give a version of this theorem for a normal operator with spectrum discrete in a rectangle area of $\mathbb C$.

                    \begin{theo} \label{bac cau}
                        Let  $A$ a normal operator and $I_1,I_2$ two intervals of
                        $\mathbb R$ such that the spectrum of $A$ in $I_1 + i I_2$ is discrete.
                        We denote the real part and the imaginary part of $A$ by $A_1$ and $A_2$. Then we have:
                         $$ \sigma(A) \cap (I_1 +i I_2) \cong \sigma_{conj}(A_1,A_2) \cap (I_1 \times I_2) . $$
                    \end{theo}

              \begin{proof}
                        It is obvious that
                            $$\sigma_{conj}(A_1,A_2) \cap (I_1 \times I_2) \subseteq \sigma(A) \cap (I_1 +i I_2).$$
                        For the inverse inclusion:
                        if $\lambda \in \sigma(A) \cap (I_1 +i I_2) $, then $\lambda $ is an eigenvalue of $A$
                        because $\sigma(A) \cap (I_1 +i I_2)$ is discrete.
                        The proposition \ref{sp ess} implies that
                        $Re(\lambda) \in  \sigma_{p}(A_1) \cap I_1$ and $Im(\lambda) \in  \sigma_{p}(A_2) \cap I_2$
                        with a nonzero common eigenvector (by equation (\ref{pt2})) and we get the inverse inclusion.
              \end{proof}

              This theorem allows us to define the monodromy of a normal pseudo-differential operator returning the quantum monodromy of the joint spectrum as below.

    \subsection{Monodromy of Normal pseudo-differential operators}  \label{mono op normal}

                In this section, we work with a space of dimension $n=2$.
                Let $P(h)$ a $h$-pseudo-differential operator on $L^2(X)$.

                We assume that $P(h)$ is normal and classical of order zero, $P(h) \in \Psi^{0}(M)$.
                As in the previous section, we can write
               $P(h)= P_1(h)+iP_2(h)$ where $P_1(h), P_2(h)$ are the real part and imaginary part of $P(h)$. \\
               The commutativity of
               $P_1(h), P_2(h)$ gives us the integrable quantum system $(P_1(h), P_2(h))$ and thanks to its joint spectrum, we can define its quantum monodromy $[\mathcal M _{qu}] \in \check{H}^1(U,GA(n, \mathbb Z) ) $ as in the previous section. Here $U$ is some open subset with compact closure in the set of regular values of the momentum map $p_0$ of principal symbols of $P_1(h), P_2(h)$, $p_0= ( Re(p), Im(p)) $ where $p$ is the principal symbol of $P(h)$.

               We assume moreover that the spectrum of $P(h)$ in $U$ is discrete. The previous theorem \label{bac cau} gives us an identification on $U$ between this spectrum and the joint spectrum.
               We have therefore the following definition of a combinatorial invariant from the discrete spectrum.

                \begin{defi}
                    With the above hypothesis, the monodromy of a normal $h$-Weyl-pseudo-differential $P(h)$ on $U$ is defined as the quantum monodromy of the integrable quantum system $(P_1(h),P_2(h))$ on $U$.

                    We call it the affine spectral monodromy.
                \end{defi}

\section{Linear Spectral Monodromy}

\subsection{Introduction}
        In this section, we propose to define the monodromy of a particular class of non-self-adjoint $ h $-pseudo-differentials operator two degree of freedom which are small perturbations of selfadjoint operators, of the form $P_\varepsilon:= P + i \varepsilon Q$ ($P$ is selfadjoint) with principal symbols $ p, q $ that commute for the Poisson bracket and in the regime $h \ll \varepsilon = \mathcal{O}(h^\delta)$ for some $0< \delta <1 $.

       The  asymptotic spectral theory by M.Hitrik-J.Sj\"{o}strand- S. Vu Ngoc (\cite{Hitrik04}, \cite{Hitrik05}, \cite{Hitrik08}...) allows us to concretely give the asymptotic expansion of eigenvalues of $P_ \varepsilon $ in a adapted complex window.

       The spectrum of $ P_ \varepsilon $ is the model of a particular discrete set which we will define in this section and call "pseudo-asymptotic lattice" (see definition \ref{pseu-cart}).
       By calculating the transition functions between the "pseudo-local cards" that are in the group $ GL(2, \mathbb Z) $,
       we can define a combinatorial invariant (the monodromy) of this lattice.
      This allows us to treat the inverse problem: define the monodromy from the spectrum of $ P_ \varepsilon $.

      We first recall some important results and analyze the general asymptotic spectral theory
      (\cite{Hitrik07}, \cite{Hitrik04}, \cite{Hitrik05}...).
      Then we will detail these results in our particular case by restating the  Birkhoff procedure of normal form.
      Next, we give some necessary steps for the construction of the monodromy of a pseudo-asymptotic lattice and then apply it to the spectrum of $ P_ \varepsilon $.

      Finally, noting that with the property of integrability, the classical monodromy (given by J.Duitermaat, \cite{Duis80}) is well defined, we also give the relationship between two monodromy types.

\subsection{Spectral Asymptotic}
    \subsubsection{Assumptions}  \label{hypothese}

    We will first give the general assumptions of our operator as in the articles \cite{Hitrik07}, \cite{Hitrik06}, \cite{Hitrik04}, \cite{Hitrik05}, \cite{Hitrik08} and the assumptions on the classical flow of the principal symbol of the non-perturbed operator and some associated spectral results: the discrete spectrum, the localization of the spectrum, the expansion of asymptotic eigenvalues ...

    $ M $ denotes $\mathbb R^2$ or a connected compact analytic real (riemannian) manifold of dimension $ 2 $ and we denote by $\widetilde{M}$ the canonical complexification of $ M $, which is either $ \mathbb C ^ 2 $ in the Euclidean case or a Grauert tube in the case of manifold (see \cite{Burns01}, \cite{Kan06}).

    We consider a non-selfadjoint $ h $-pseudo-differential operator $P_{\varepsilon}$ on $ M $ and suppose that
    \begin{equation}
         P_{\varepsilon=0}:= P \quad \textrm{is formally self-adjoint}.
    \end{equation}

    Note that if $ M = \mathbb R ^ 2 $, the volume form $ \mu (dx) $ is naturally induced by the Lebesgue measure on $ \mathbb R ^ 2 $, contrariwise in the case $ M $
    is compact riemannian manifold, the volume form $ \mu(dx) $ is induced by the given riemannian structure of $ M $.
    So, in all cases, the operator $P_{\varepsilon}$ is seen as an (unbounded) operator  on $L^2(M, \mu(dx)) $.

    We always denote the principal symbol of $ P_ {\varepsilon} $ by $ p_ \varepsilon $ which is defined on $ T ^ * M $ as we discussed in previous section.
    ~~\\
    We'll assume the ellipticity condition at infinity for $ P_ {\varepsilon} $ as follows:

    When $M=\mathbb R^2$, let
                        \begin{equation}
                            P_{\varepsilon}= P(x,hD_x,\varepsilon; h )
                        \end{equation}
   be the Weyl quantification of a total symbol $P(x, \xi,\varepsilon; h )$ depending smoothly on $\varepsilon$ in a neighborhood of $(0, \mathbb R) $ and taking values in the space of holomorphic functions of $(x,\xi)$ in a tubular neighborhood of $\mathbb R^4$ in $\mathbb C^4$ on which we assume that:
                \begin{equation}
                           | P(x, \xi,\varepsilon; h ) | \leq  \mathcal O(1) m(Re(x,\xi)).
                \end{equation}

    Here $ m $ is an order function in the sense of definition \ref{fonc ord}.
    We assume moreover that $ m> 1 $ and $ P_ {\varepsilon} $ is classical
     \begin{equation}
        P(x, \xi,\varepsilon; h ) \sim \sum_{j=0}^\infty
                             p_{j,\varepsilon}(x,\xi) h^j, h \rightarrow 0,
     \end{equation}
     in the selected space of symbols. \\
     In this case, the main symbol is the first term of the above expansion, $p_\varepsilon = p_{0,\varepsilon}$ and the ellipticity condition at infinity is
                        \begin{equation}
                            |p_{0,\varepsilon}(x,\xi)| \geq \frac{1}{C} m(Re(x,\xi)), \mid (x,\xi)\mid \geq C,
                        \end{equation}
     for some $ C> 0 $ large enough.

     When $ M $ is a manifold, we consider $ P_ \varepsilon $ a differential operator on $ M $ such that in local coordinates $ x $ of $ M $, it is of the form:
        \begin{equation}
                            P_\varepsilon = \sum_{|\alpha |\leq m} a_{\alpha,\varepsilon}(x;h)(hD_x)^\alpha,
        \end{equation}
   Where $D_x= \frac{1}{i} \frac{\partial}{\partial x}$ and $a_{\alpha,\varepsilon}$ are smooth functions of $\varepsilon$ in a neighborhood of $0$
   with values in the space of holomorphic functions on a complex neighborhood of $x=0$. \\
  We assume that these $a_{\alpha,\varepsilon}$  are classic
                         \begin{equation}
                           a_{\alpha,\varepsilon}(x;h) \sim \sum_{j=0}^\infty
                             a_{\alpha,\varepsilon,j}(x) h^j, h \rightarrow 0,
                        \end{equation}
  in the selected space of symbols. \\
  In this case, the principal symbol $p_\varepsilon$ in the local canonical coordinates associated $(x,\xi) $ on $T^*M $ is \begin{equation}
                           p_\varepsilon(x,\xi)= \sum_{|\alpha | \leq m} a_{\alpha,\varepsilon,0}(x) \xi^{\alpha}
  \end{equation}
  and the elipticity condition at infinity is
            \begin{equation}
                            |p_{\varepsilon}(x,\xi)| \geq \frac{1}{C} \langle \xi \rangle ^m, (x,\xi) \in T^*M, \mid \xi \mid \geq C,
            \end{equation}
  for some $C>0$ large enough.\\
  Note here that $ M $ has a riemannian metric, then $\mid \xi \mid$ and $\langle \xi \rangle= (1+  \mid \xi \mid ^2)^{1/2}$ is well defined.

  It is known from articles \cite{Hitrik07}, \cite{Hitrik04} that with the above conditions, the spectrum of $ P_ \varepsilon $ in a small but fixed neighborhood of $ 0 \in \mathbb C $ is discrete,
  when $h>0, \varepsilon \geq 0$ are small enough. Moreover, this spectrum is contained in a band of size $\varepsilon$:
                            \begin{equation}|Im(z)| \leq \mathcal O(\varepsilon ).\end{equation}
  This gives the first location of the spectrum of $ P_ \varepsilon $.

Let $p= p_{\varepsilon=0}$, it is principal symbol of the selfadjoint unperturbed operator $ P $ and therefore
real. \\
 We assume that
                        \begin{equation}p^{-1}(0) \cap T^*M  \qquad \textrm{ is connected} \end{equation}
and the energy level $ E = 0 $ is regular for $ p $, i.e  $dp \neq 0$ on $p^{-1}(0) \cap T^*M$.

Let $q=\frac{1}{i}(\frac{\partial}{\partial \varepsilon})_{\varepsilon
                    =0}p_\varepsilon$, so
                            \begin{equation}  \label{symb prin}
                                  p_\varepsilon=p+i \varepsilon q+ \mathcal O (\varepsilon ^2)
                            \end{equation}
in the neighborhood of $p^{-1}(0) \cap T^*M $.

For $T>0$, we introduce the symmetric average time $T$ of $q$ along the flow of $p$, defined near $p^{-1}(0) \cap T^*M $:
                            \begin{equation}  \label{t-average}
                                 \langle q \rangle _T= \frac{1}{T} \int_{-T/2}^{T/2} q \circ exp(t H_p) dt,
                            \end{equation}
where $H_p= \frac{\partial p}{\partial \xi} \cdot \frac{\partial }{\partial x}- \frac{\partial p}{\partial x} \cdot
\frac{\partial }{\partial \xi}$ is the hamiltonian vector field of $p$. \\
Note that $H_p(\langle q \rangle _T)= \{p,\langle q \rangle_T\}= \mathcal O (1/T)$.

As explained in \cite{Hitrik07}, by introducing a Fourier integral operator
 (which is defined microlocally close to $p^{-1}(0) \cap T^*M$ ),
 we can reduce our operator to a new operator, denoted again by  $P_\varepsilon$, with principal symbol
 $ p+i\varepsilon \langle q \rangle _T+ \mathcal O (\varepsilon ^2) $ and $P_{\varepsilon=0}$ is still the original unperturbed operator.
 So we can assume that our operator $P_\varepsilon$ is microlocally defined in the neighborhood of $p^{-1}(0) \cap T^*M $ with $ h $- principal symbol
                     \begin{equation}
                                  p+i \varepsilon \langle q \rangle _T+ \mathcal O (\varepsilon ^2).
                     \end{equation}

  Consequently, with the help of the sharp Garding inequality the spectrum of $P_\varepsilon$ in the domain $ \{z \in \mathbb C : |z| < \delta \}$, when $\varepsilon, h, \delta \rightarrow 0$ is confined in the band (voir \cite{Sj00.1}, \cite{Sj00.2}):
                         \begin{equation} \label{loca spectre}
                             ]-\delta, \delta[ + i \varepsilon \big [ \lim_{T\rightarrow \infty} \inf_{p^{-1}(0)}
                             Re \langle q \rangle _T -o(1),
                             \lim_{T\rightarrow \infty} \sup_{p^{-1}(0)} Re \langle q \rangle _T + o(1)  \big ].
                          \end{equation}

With more assumptions about the dynamics of classical flow of the first term of the unperturbed symbol (in
a certain energy level), one can obtain more detailed results on the asymptotic distribution of eigenvalues in
such a band. \\
 Let a given value $F_0 \in \big [\lim_{T\rightarrow \infty}
                      \inf_{p^{-1}(0)} Re \langle q \rangle _T, \lim_{T\rightarrow \infty}
                       \sup_{p^{-1}(0)} Re \langle q \rangle _T
                     \big ] $,
we want to determine all the eigenvalues of $P_\varepsilon$ in a rectangle of center $(0, \varepsilon F_0)$ and of size $\mathcal{O}(h^\delta) \times \mathcal{O }(\varepsilon h^\delta)$
(which is included in the previous band) for
                             $$h\ll \varepsilon \leq \mathcal O(h^\delta),$$
where $\delta>0$ is some number small enough but fixed.

 \begin{rema}
    The problem of determining asymptotically the eigenvalues of $ P_ \varepsilon $ in such a rectangle of spectral domain has been proposed in the literature with different assumptions on the Hamiltonian flow of $ p $: this flow can be periodic on an energy $ E $ near $ 0 $, completely integrable or almost integrable.

    The force of the perturbation $\varepsilon$ is treated with several regimes:
     $h^M \leq \varepsilon \leq \mathcal O(h^\delta)$, for $M $ fixed large enough,
     $h\ll \varepsilon \leq \mathcal O(h^\delta)$, $h^{1/3- \delta} < \varepsilon \leq \varepsilon_0$,...
    and the size of the rectangle: depends on $ h $ or does not depend on $ h $.

   One can read the articles \cite{Hitrik04}, \cite{Hitrik05}, \cite{Hitrik07},
                              \cite{Hitrik08}, \cite{Sj09}. \\
   Here, we present the completely integrable case in the regime $h\ll \varepsilon \leq \mathcal O(h^\delta)$.
   \end{rema}

Now, assume that $ p $ is completely integrable in a neighborhood of $p^{-1}(0) \cap T^*M$,
i.e there exists a smooth real function $f$, independent of $p$ such that $ \{p, f \}= 0$.
As explained in (\cite{Hitrik07}, page 21-22 and 55), the energy space $p^{-1}(0)$ is decomposed into a singular foliation:
 \begin{equation} p^{-1}(0) \cap T^*M   = \bigcup_{a \in J} \Lambda_a ,\end{equation}
 where $ J $ is assumed to be a connected graph with a finite number of vertices and of edges. We denote by $S$ the set of vertices.

For each $a \in J$, $\Lambda_a$ is a connected compact subset invariant with respect to $H_p$. Moreover, if $a \in J\backslash S$,
                      $\Lambda_a$ is a invariant Lagrangian torus depending analytically on $a$.
Each edge of $ J $ can be identified with a bounded interval of $ \mathbb R $.

Next, we assume the continuity of tori:
let $\Lambda_{a_0}, a_0 \in  J\backslash S$, for all $\mu >0, \exists \gamma >0$, such that
if $dist(a,a_0) < \gamma $, then $\Lambda_a \subset \{\rho \in p^{-1}(0) \cap T^*M: dist(\rho, \Lambda_0) < \mu \}$.
Note that this hypothesis holds for integrable systems with non-degenerate singularities.

For each torus $\Lambda_a, a  \in J\backslash S$, by the action-angle theorem \ref{A-A} there are
action-angle local coordinates
                        $(x, \xi) $  near $\Lambda_a$ such that $\Lambda_a \simeq \{ \xi=0\}$ and that $p$
becomes a function of $\xi$, $p=p(\xi)= p(\xi_1, \xi_2)$.
The frequency of $\Lambda_a$ can be defined as an element of the real projective line by
                      \begin{equation} \label{frequence} \omega(a)= [p_{\xi_1}'(0):p_{\xi_2}'(0)] .\end{equation}
Sometimes $\omega(a)$ is seen as an element of $\mathbb R$. \\
Moreover, by the action-angle theorem,  we know that $\omega(a)$ depends analytically of $a  \in J\backslash S$.
We will assume that the function $a \mapsto \omega(a) $ is not identically constant on any connected component of $J \backslash S$.

For each $a \in J$, we define a compact interval in $\mathbb R$:
                            \begin{equation} \label{Q vo cung}
                                    Q_\infty(a)=
                                     \big [ \lim_{T\rightarrow \infty} \inf_{\Lambda_a} Re \langle q \rangle _T,
                                      \lim_{T\rightarrow \infty} \sup_{\Lambda_a} Re \langle q \rangle _T\big].
                            \end{equation}
Then the spectral localization (\ref{loca spectre}) becomes
                      \begin{equation}  \label{loca. spectre 2}
                            Im(\sigma(P_\varepsilon) \cap \{z \in \mathbb C: |Re z| \leq \delta \}) \subset
                             \varepsilon \big [ \inf \bigcup_{a \in J}Q_\infty(a)-o(1),
                              \sup \bigcup_{a \in J}Q_\infty(a) +  o(1) \big ],
                      \end{equation}
when $\varepsilon, h, \delta \rightarrow 0$.

From now, for simplicity, we will assume that $q$ is real. \\
For each torus $\Lambda_a, a  \in J\backslash
S$, one defines $\langle q  \rangle_{\Lambda_a}$ the average of $ q $ with respect to the natural Liouville
measure on $\Lambda_a$
                      \begin{equation} \label{moyenne de q}
                       \langle q \rangle_{\Lambda_a}= \int_{\Lambda_a}q \end{equation}

\begin{rema}
 In action-angle coordinates $(x,\xi)$ near $\Lambda_a$ such that $\Lambda_a \simeq \{ \xi=0\}$, we have
                                    \begin{equation} \label{moyenne2}
                                     \langle q \rangle (\xi)=
                                     \frac{1}{(2\pi)^2}\int_{\mathbb{T}^2}q(x,\xi)dx.
                                    \end{equation}
In particular, $\langle q \rangle_{\Lambda_a}=\langle q \rangle(0)$.
\end{rema}

\begin{rema}[(\cite{Hitrik07}, page 56-57)]  \label{rem1}
                            For $a \in J\backslash S$:

                            - if $\omega(a) \notin \mathbb{Q}$, then $Q_\infty(a)= \{\langle q
                            \rangle_{\Lambda_a}\}$.

                            - if $\omega(a)= \frac{m}{n} \in \mathbb{Q}$ ($m \in \mathbb Z,  n \in \mathbb N$), then
                            $$Q_\infty(a)= \langle q \rangle_{\Lambda_a} + \mathcal O \big(\frac{1}{k(\omega(a))^\infty}  \big)[-1,1], \quad k(\omega(a)):= |m|+|n|.$$

In particular $$\sum_{a:\omega(a) \in \mathbb{Q} } |Q_\infty(a) | < \infty.$$
\end{rema}

$\langle q \rangle_{\Lambda_a} $ depends analytically of $a \in J\backslash S$ and we assume it can be extended continuously on $J$.
Furthermore, we assume that the function $a \mapsto  \langle q \rangle (a) =\langle q \rangle_{\Lambda_a} $ is not identically constant on any connected component of $J\backslash S$.

Note that $p$ and $\langle q \rangle $ commute in neighborhood of $p^{-1}(0) \cap T^*M$.

\subsubsection{Asymptotic eigenvalues}

 \begin{defi}  \label{diop}

 For a torus $\Lambda_a,$ $ a  \in J\backslash S$ and $\omega(a)$ defined as (\ref{frequence}) and
 let $\alpha >0 $, $d>0$, we say that $\Lambda_a$ is $(\alpha,d)-$Diophantine if:
                            \begin{equation}  \label{dn alpha-d dioph}
                                    \big | \omega(a)- \frac{m}{n}\big | \geq \frac{\alpha}{ n^{1+d}}, \quad
                                    \forall  m \in \mathbb Z, n  \in \mathbb N^*,
                            \end{equation}
    here $\omega(a)$ is seen as an element of $\mathbb R$.
 \end{defi}

  Note also that when $d>0$ is fixed, the Diophantine property (for some $\alpha$) of $\Lambda_a$ is independent of the choice of action-angle coordinates.

\begin{defi}  \label{dn bonnes valeurs}
        For $\alpha>0 $ and $d>0$, we define the set of "good values" $\mathcal{G}(\alpha,d)$  obtained from $\cup_{a \in J}Q_\infty(a)$ by removing the following set of "bad values" $\mathcal{B}(\alpha,d)$:
        \begin{displaymath}
                                        \mathcal{B}(\alpha,d)= \Bigg ( \bigcup_{dist(a,S) < \alpha} Q_\infty(a) \Bigg )
                                          \bigcup \Bigg ( \bigcup_{a \in J\backslash S: |\omega'(a)| < \alpha } Q_\infty(a) \Bigg )
                                           \bigcup \Bigg( \bigcup_{a \in J\backslash S: |d \langle q \rangle_{\Lambda_a}| < \alpha }
                                            Q_\infty(a) \Bigg )
        \end{displaymath}
                                \begin{displaymath}
                                    \bigcup \Bigg ( \bigcup_{a \in J\backslash S: \omega(a) \textrm{ is not }
                                    (\alpha,d)-
                                      \textrm{Diophantine}}
                                     Q_\infty(a)\Bigg ).
                                \end{displaymath}
\end{defi}

\begin{rema} \label{rem2}

         \begin{itemize} \label{peti mesure}
                                \item The measure of the set of bad values $\mathcal{B}(\alpha,d)$ in
                                    $\cup_{a \in J}Q_\infty(a)$ is small ($\mathcal O (\alpha)$) when $\alpha
                                    >0$ is small and $d>0$ is fixed, provided that the measure of
                                    \begin{equation}  \label{a condition}
                                            \Bigg ( \bigcup_{a \in
                                             J\backslash S: \omega(a)| \in \mathbb Q } Q_\infty(a)
                                             \Bigg ) \bigcup \Bigg ( \bigcup_{a \in S} Q_\infty(a)
                                            \Bigg )
                                    \end{equation}
                                    is sufficiently small, depending on $\alpha$ (see \cite{Hitrik07}).
                                \item If $F_0 \in \mathcal{G}(\alpha,d) $ is a good value , then by definition of $\mathcal{B}(\alpha,d) $ and remark (\ref{rem1}), the pre-image $\langle q \rangle^{-1}(F_0) $ is a finite set
                                         $$\langle q \rangle ^{-1}(F_0)= \{a_1,\ldots,a_L\} \subset J\setminus S.$$
                                    The corresponding tori $\Lambda_{a_1},\ldots,\Lambda_{a_L}$ are tori
                                    $(\alpha, d)$-Diophantine of $p^{-1}(0) \cap T^*M$.
                                    By this way, when
                                    $F_0 $ varies in $\mathcal{G}(\alpha,d) $, we obtain a Cantor family of invariant tori $(\alpha, d)$-Diophantine in the energy space
                                    $p^{-1}(0) \cap T^*M$.
         \end{itemize}
\end{rema}



            \begin{defi}[\cite{lectureColin}, \cite{Arnold67}, \cite{Cappel94}]
                        Let $ E $ is a symplectic space and his Lagrangian Grassmannian $ \Lambda (E) $ (which is set of all Lagrangian subspaces of $ E $). We consider a bundle $ B $ in $ E $ over the circle or a compact interval provided with a Lagrangian subbundle called vertical.
                        Let $ \lambda (t) $ a section of $ \Lambda (B) $ which is transverse to the vertical edges of the interval in the case where the base is an interval.

                        The Maslov index of $ \lambda (t) $ is the intersection number of this curve with the singular cycle of Lagrangians which do not cut transversely the vertical subbundle.
            \end{defi}

\begin{theo}[\cite{Hitrik07}]  \label{theorem quasi-spectre}
    Suppose that $P_\varepsilon$ is an operator with principal symbol (\ref{symb prin}) and satisfying the above conditions.
    Let $F_0 \in  \mathcal{G}(\alpha,d)$ a good value. As in the remark (\ref{rem2}), we write
    $ \langle q \rangle^{-1}(F_0)= \{a_1,\ldots,a_L\} \subset J\setminus S $ and the corresponding tori $\Lambda_{a_1},\ldots,\Lambda_{a_L}$ in $p^{-1}(0) \cap T^*M$. \\
     For each $j=1, \ldots, L $, note $S_j \in \mathbb R^2$ the action and $k_j \in \mathbb Z^2$ the Maslov index of the fundamental cycles $(\gamma_{1,j},\gamma_{2,j})$ of $\Lambda_{a_j}$  which are defined by
            $$\kappa_j(\gamma_{l,j}) =  \{ x \in \mathbb{T}^2: x_l=0 \}, l=1,2,$$
    where $\kappa_j$ is a action-angle coordinates in neighborhood of torus $\Lambda_{a_j}$,
    \begin{equation}  \label{coor}
    \kappa_j: (\Lambda_{a_j}, T^*M) \rightarrow (\xi=0, T^* \mathbb T ^2)
    \end{equation}
    We assume that $h \ll \varepsilon = \mathcal{O}(h^\delta)$ for $0< \delta <1 $.

    Then the eigenvalues of $ P_ \varepsilon $ with multiplicity in a rectangle of form
     \begin{equation}  \label{cua so}
             R(\varepsilon,h)= \Big[-\frac{h^\delta}{\mathcal{O}(1)},\frac{h^\delta}{\mathcal{O}(1)} \Big]
                                        +i \varepsilon \Big [ F_0-\frac{h^\delta}{\mathcal{O}(1)}, F_0+ \frac{h^\delta}{\mathcal{O}(1)} \Big ]
      \end{equation}
    are given by
                                    \begin{equation} \label{eigenvalues}
                                        P_j^{(\infty)} \Big( h(k-\frac{k_j}{4})-\frac{S_j}{2 \pi},\varepsilon; h\Big) + \mathcal O(h^\infty),
                                        k \in \mathbb Z^2, 1 \leq j \leq L.
                                    \end{equation}
    Here $P_j^{(\infty)}(\xi,\varepsilon; h )$ is a smooth function of $\xi$ in a neighborhood of $(0, \mathbb R^2)$ and $\varepsilon$ in a neighborhood of $(0, \mathbb R)$, real valued for $\varepsilon =0$ and admits an asymptotic expansion in the space of symbols.
    \begin{equation} \label{symbole normal}
                                        1 \leq j \leq L, \quad P_j^{(\infty)} ( \xi,\varepsilon; h) \sim \sum_{k=0}^{\infty}
                                                        h^k p_{j,k}^{(\infty)} (\xi, \epsilon)
                                    \end{equation}
                                    whose principal symbol is
                                     \begin{equation} \label{prin normal}
                                             p_{j,0}^{(\infty)} (\xi, \varepsilon)= p_j(\xi)+ i \varepsilon \langle q_j \rangle (\xi) +
                                             \mathcal O(\varepsilon ^2).
                                    \end{equation}
    Here $p_j, q_j$ are the expressions of $p,q$ in action-angle variables near of $\Lambda_j$, given by (\ref{coor}) and  $\langle q_j \rangle $is the average of $q_j$ on tori, defined in (\ref{moyenne2}).
\end{theo}

 \begin{rema}
      In the case of the above theorem that for every $j =1, \ldots,L$, the eigenvalues form a deformed spectrum lattice in the rectangle (\ref{cua so}) of size $(h^\delta  \times \varepsilon h^\delta) $.
      Therefore the spectrum of $ P_ \varepsilon $ in the rectangle therefore is the union of $ L $ such lattices.

      Note that this is not valid for every rectangle. However, it is valid for a "good rectangle" whose center $(0, \varepsilon F_0)$ with $ F_0 $ is a good value.
      However, as we said in the remark \ref{rem2}, with the condition (\ref{a condition}), the complement of the set of good values is a small measure (see \cite{Poschel82}),
      then there are many such good rectangles in the band (\ref{loca. spectre 2}). This signifies that one can give asymptotically "almost all" eigenvalues of $ P_ \varepsilon $ in this band.
\end{rema}

    \begin{rema}
            In the case where $p$ is nearly integrable, the result of the theorem is
                                 still true thanks to the existence invariant KAM tori
           which allows us to realize microlocally the construction of the quantum normal form of $P_\varepsilon$ (see section 7.3 in \cite{Hitrik07}).

         For the KAM theory (Kolmogorov-Arnold-Moser), one may consult the references
         \cite{Poschel82}, \cite{Broer07}, \cite{Bost86}, \cite{Broer91}, \cite{Delshams96}.

     \end{rema}

 \begin{proof}[Main idea of the proof of theorem(\ref{theorem quasi-spectre})]

     For a detailed proof of the theorem, one can consult \cite{Hitrik07}, \cite{Hitrik06}.
     We will give here some important ideas of the proof of theorem.

     The principle is the formal construction of the Birkhoff quantum normal form for $P_\varepsilon$, microlocally near a fixed Diophantine torus in $p^{-1}(0) \cap T^*M$, say $\Lambda_1 \in \{ \Lambda_{a_1},\ldots,\Lambda_{a_L} \}$.
     The Diophantine condition is necessary for this construction. For this method, see also \cite{Ali85}, \cite{Popov00},\cite{Bambusi99}, \cite{Eckhardt86}.

       In this procedure we first use (formally) a canonical (symplectic) transformation for the total symbol of $P_\varepsilon $ in order to reduce it to the normal form (\ref{symbole normal}),(\ref{prin normal})
      modulo $\mathcal O(h^\infty)$ which is independent of $ x $ and homogeneous in $(h, \xi, \varepsilon)$ in all orders.
      Then, the operator $P_\varepsilon$ is conjugated by a Fourier integral operator with complex phase to a new operator with such a total symbol. \\
       Indeed, by introducing action-angle coordinates near $ \Lambda_1 $, $ P_ \varepsilon $ is  microlocally defined around the section $\xi=0$ in $T^* \mathbb T^2$  and its principal symbol (\ref{symb prin}) has the form:
                             \begin{equation}  \label{symbol micro}
                                  p_\varepsilon(x,\xi)=p(\xi) +i \varepsilon q(x,\xi)+ \mathcal O (\varepsilon ^2)
                            \end{equation}
       with $p(\xi)= \widetilde{\omega} \cdot \xi + \mathcal O (\xi^2)$,
       where $\widetilde{\omega}= (p_{\xi_1}'(0),p_{\xi_2}'(0))$ and the frequency $ \omega(a)= [ p_{\xi_1}'(0):p_{\xi_2}'(0)]$, defined in (\ref{frequence}) satisfies the condition (\ref{dn alpha-d
                        dioph}).

       Then, by the Birkhoff normal form procedure, for any arbitrary fixed $ N $ large enough, we can construct a holomorphic canonical transformation $\kappa_\varepsilon^{(N)}$ defined in a complex neighborhood of $ \xi = 0 $ in $ T ^* \mathbb T ^ 2 $ such that the total symbol $ P $ of $ P_ \varepsilon $ is reduced to a new symbol:
                        \begin{equation}
                             P \circ \kappa_\varepsilon^{(N)} (x, \xi,\varepsilon; h )= p_0+ h p_1+ h^2 p_2+\cdots,
                        \end{equation}
       where every $p_j= p_j(x,\xi,\varepsilon), j\geq1$ holomorphic near $\xi=0$ in $T^* \mathbb T^2$, depending smoothly in $\varepsilon \in (0, \mathbb R)$, independent of $x$ to order $N$ and it is important that the principal symbol $ p_0 $ satisfies
       \begin{equation}
                         p_0= p_\varepsilon \circ \kappa_\varepsilon^{(N)}  (x,\xi)= p^{(N)}(\xi, \varepsilon) + r_{N+1}(x,\xi,\varepsilon),
       \end{equation}
       where $p^{(N)}(\xi, \varepsilon)= p(\xi)+ i \varepsilon \langle q \rangle (\xi) +
                                             \mathcal O(\varepsilon ^2)$, $\langle q \rangle (\xi)$
       given by (\ref{moyenne2}), $r_{N+1}(x,\xi,\varepsilon) = \mathcal O((\xi, \varepsilon )^{N+1})$. \\
    Thus, $ p_0 $ has the same form as (\ref{prin normal}).

     On operator level, $P_\varepsilon$ is conjugated to a new operator of the form
                        \begin{equation}  \label{op reduit}
                            P^{(N)}(h D_x, \varepsilon; h)+ R_{N+1}(x,h D_x, \varepsilon; h ),
                        \end{equation}
        where $P^{(N)}(h D_x, \varepsilon; h)$ has a total symbol  independent
        of $x$ whose principal symbol is $p^{(N)}$ and $R_{N+1}(x,\xi, \varepsilon; h )=\mathcal O((h,\xi, \varepsilon )^{N+1}) $. \\
        The operator (\ref{op reduit}) acts on the space $L_\theta^2(\mathbb T^2)$
        of Floquet periodic functions microlocally defined over $\mathbb T^2$ whose an element $u$ satisfies
                        $$u(x-\nu)= e^{i \theta\cdot \nu } u(x), \quad \theta= \frac{S_1}{2\pi h}+ \frac{k_1}{4}, \quad \nu \in 2\pi \mathbb Z^2.$$
         An orthonormal basis of this space is
         $$ \Big \{x \in \mathbb T^2,  e_k(x)= e^{ix(k-\theta)}=
                           e^{\frac{i}{h}x.\big (h(k-\frac{k_1}{4})-\frac{S_1}{2\pi } \big )}, k \in \mathbb Z^2 \Big \}.$$
         Consequently, the eigenvalues of $ P_ \varepsilon $ modulo $ \mathcal O (h ^ \infty) $ are given by (\ref{eigenvalues}).
    \end{proof}

   \begin{rema} \label{bien doi sym prin}
        For all $j=1, \ldots, L$, from (\ref{prin normal}), at $\xi=0$ we have
        $ p_{j,0}^{(\infty)} (0,\varepsilon)=  i \varepsilon F_0 + \mathcal O(\varepsilon ^2)$ and therefore
                    $ P_j^{(\infty)}(0,\varepsilon;h)= i \varepsilon F_0 +
                                 \mathcal O(\varepsilon ^2)+ \mathcal O(h) $.
        Consequently, $p_{j,0}^{(\infty)} (0,\varepsilon) \sim i \varepsilon F_0$
        when $\varepsilon \rightarrow 0 $ and $ P_j^{(\infty)}(0,\varepsilon;h) \sim i \varepsilon F_0  $ when $\varepsilon,h \rightarrow 0 , h \ll \varepsilon $.

        Moreover, we have also $ d (p_j)_{| \xi= 0 }=: a_j = (a_{1,j},a_{2,j}) \in \mathbb R^2 $ and
        $ d ( \langle q_j \rangle)  _{| \xi= 0}=:b_j= (b_{1,j},b_{2,j}) \in \mathbb R^2 $ are $\mathbb R$-linearly independent.
        We can rewrite the principal symbol (\ref{prin normal}) in the form
                         \begin{equation}  \label{prin normal da bien doi}
                            p_{j,0}^{(\infty)} (\xi, \epsilon)= i \varepsilon F_0+ (a_j+i \varepsilon b_j)\cdot \xi+  \mathcal O(\xi ^2)
                                             + \mathcal O(\varepsilon ^2).
                         \end{equation}
   \end{rema}

        \begin{prop}  \label{diffeo}
            Let $\lambda= P(\xi; \varepsilon, h)$ a complex-valued smooth function of $\xi$ near $0 \in \mathbb R^2$ and of small parameters $ h, \varepsilon $ near $0 \in \mathbb R$.
             Suppose that we can write $P$ in the form
             $$P(\xi; \varepsilon, h)= P_0+ \mathcal O(h)$$
             with
             $$P_0= P_0(\xi; \varepsilon)= g_1(\xi)+ i \varepsilon g_2(\xi) + \mathcal O(\varepsilon ^2)$$
             such that $dg_1(0) \wedge dg_2(0) \neq 0$.

             If we assume that $h \ll \varepsilon$, then for $h$ et $\varepsilon$ small enough,
             there are $\rho, r>0$ small enough such that $P$ is a local diffeomorphism near $\xi=0$ from $B(0, \rho)$
             to its image, denoted $B(\varepsilon)$.
         \end{prop}

         \begin{proof}
            First, seeing $P$ as a function of $\mathbb R^2$, we set $\widehat{P}:= \chi^{-1} \circ P$.
            Then we can write
                $$\widehat{P}= g_1(\xi)+ i g_2(\xi) + \mathcal O(\varepsilon) + \mathcal O(\frac {h}{\varepsilon}). $$
            Let $a=(a_1,a_2)= dg_1(0), b=(b_1,b_2)=  dg_2(0) $ and $M=\mid a_1b_2-a_2 b_1 \mid > 0$. \\
            The differential of $\widehat{P}$ in $\xi=0$ is
             \begin{equation*}
                            \frac{\partial \widehat{P} } {\partial \xi}(0)
                            = a+i \varepsilon b + \mathcal{O}(\varepsilon ) + \mathcal O(\frac{h}{\varepsilon} )=
                                \left( \begin{array}{ccc}
                                a_1+\mathcal{O}(\varepsilon ) + \mathcal O(\frac{h}{\varepsilon} )&
                                  a_2 + \mathcal{O}(\varepsilon ) + \mathcal O(\frac{h}{\varepsilon} ) \\
                                 b_1+ \mathcal{O}(\varepsilon ) + \mathcal O(\frac{h}{\varepsilon} )
                                 &  b_2+ \mathcal{O}(\varepsilon ) + \mathcal O(\frac{h}{\varepsilon} )
                                 \end{array} \right)
              \end{equation*}
              and thus
              $$\mid det( \frac{\partial \widehat{P}}{\partial \xi} (0)) \mid
                          = M+ \mathcal{O}(\varepsilon) + \mathcal O(\frac{h}{\varepsilon} ) .$$
              Then for $h, \varepsilon$ small enough and $h \ll \varepsilon$, it's clear that
              $\mid det( \frac{\partial \widehat{P}}{\partial \xi} (0)) \mid  \simeq M  $ is nonzero.
              Therefore, the local inverse function theorem ensures that
              $\widehat{P}$ is local diffeomorphism in $\xi=0$.
              Hence we get the desired result for $P$.
         \end{proof}


        Let us return to the spectral problem of $ P_\varepsilon$ discussed in the theorem (\ref{theorem quasi-spectre}).
        For each $j=1, \ldots,L$, as an application of the previous lemma with $P= P_j^{(\infty)}$, then
        $ P_j^{(\infty)}$ is a smooth local diffeomorphism in $\xi=0 \in \mathbb R^2$ from a neighborhood of $0$ to its image, noted by $B_j(\varepsilon)$ .
        Note that for $h$ small enough, the good rectangle $R(\varepsilon,h)$ is always included in $B_j(\varepsilon)$.

        We denote $\Sigma_j(\varepsilon,h) \subset R(\varepsilon,h) $ the quasi-eigenvalues of $ P_\varepsilon$ in $R(\varepsilon,h)$, given by the image by $P_j^{(\infty)}$ of
        $ \xi= h(k-\frac{k_j}{4})-\frac{S_j}{2 \pi}, k \in \mathbb Z^2$.

        Writing $hk= \xi + h \frac{k_j}{4} +  \frac{S_j}{2 \pi}$ and letting
                         \begin{equation} \label{hk}
                                f_j:= (P_j^{(\infty)})^{-1}+ h \frac{k_j}{4} +  \frac{S_j}{2 \pi},
                         \end{equation}
        then $f_j=f_j(\lambda, \varepsilon;h)$ is a local diffeomorphism from
                         $B_j(\varepsilon)$ to its image.
        Denote $E_j(\varepsilon,h)=f_j(R(\varepsilon,h))
                         $ which is close to $ \frac{S_j}{2 \pi}$ and
                         $\Gamma_j(\varepsilon,h):=f_j(\Sigma_j(\varepsilon,h))$, then we have
                         $\Gamma_j(\varepsilon,h)=h \mathbb Z^2 \cap  E_j(\varepsilon,h).$ \\
        In summary, we have:
                         \begin{eqnarray} \label{new window}
                                         f_j: & R(\varepsilon,h) \rightarrow E_j(\varepsilon,h) \\
                                  f_j\mid_{\Sigma_j(\varepsilon,h)}:& \Sigma_j(\varepsilon,h) \rightarrow
                                      \Gamma_j(\varepsilon,h) \subset h \mathbb Z^2
                         \end{eqnarray}

                \begin{rema}
                       On the other hand, if we assume that $L=1$, the theorem (\ref{theorem quasi-spectre})
                       asserts that in $R(\varepsilon,h)$, the quasi-eigenvalues are equal to the real eigenvalues of $P_\varepsilon$ modulo $\mathcal O(h^\infty)$:
                                  \begin{equation}
                                      \sigma(P_\varepsilon) \cap R(\varepsilon,h) = \Sigma_1(\varepsilon,h)+ \mathcal O(h^\infty),
                                  \end{equation}
                       in the sense that there is a bijection
                                 \begin{equation}
                                     \chi:\Sigma_1(\varepsilon,h)\rightarrow \sigma(P_\varepsilon) \cap R(\varepsilon,h)
                                 \end{equation}
                      such that $\chi= Id + \mathcal O(h^\infty)$.
                      The diffeomorphism $f:= f_1$ in (\ref{new window}) thus satisfies
                                 \begin{eqnarray} \label{axa spectre}
                                             f :  R(\varepsilon,h) & \rightarrow & E_1(\varepsilon,h)
                                                        \nonumber \\
                                       \sigma(P_\varepsilon) \cap R(\varepsilon,h) \ni  \lambda & \mapsto &
                                               f(\lambda,\varepsilon; h) \in h \mathbb Z^2 +\mathcal O(h^\infty).
                                \end{eqnarray}

 				\begin{figure}[!h]
 \begin{center}
                            \includegraphics[width=0.8\textwidth]{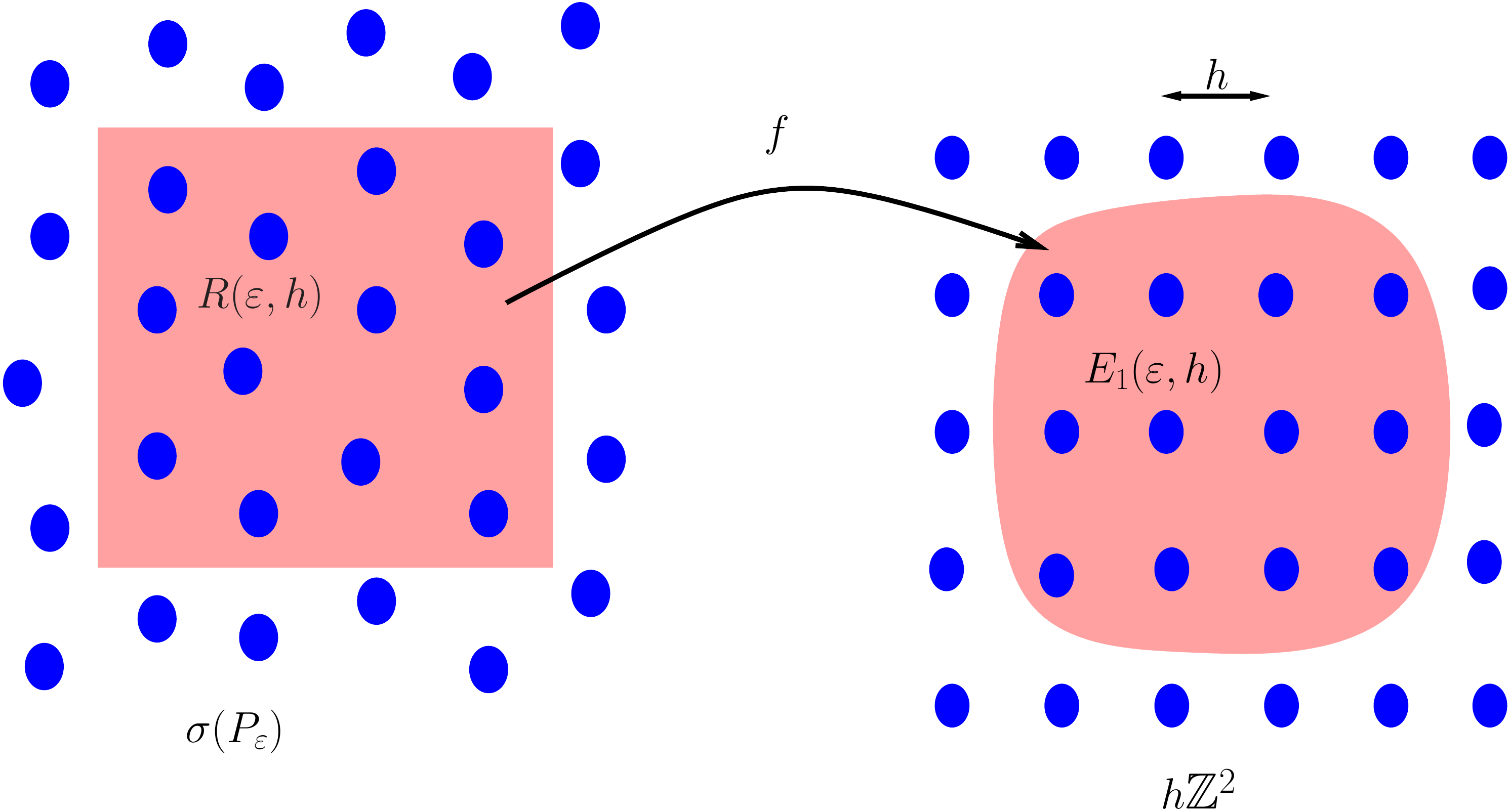} \\
                             \caption{A micro-chart of spectrum of $P_\varepsilon$ }
 \end{center}
                \end{figure}
 			
~~\\
  In particular, we have a bijection between the sets
                               \begin{equation} \sigma(P_\varepsilon) \cap R(\varepsilon,h) \simeq
                                         \Sigma_1(\varepsilon,h) \simeq \Gamma_1(\varepsilon,h) \subset h \mathbb Z^2
                               \end{equation}
        \end{rema}

            \subsubsection{What is the size of $E(\varepsilon,h)$?}  \label{la taille de E}

                    As we know, the surface $R(\varepsilon,h)$ is of size
                         $\mathcal{O}(h^\delta) \cdot \mathcal{O}(\varepsilon h^\delta)
                         $, Now we are interested the size of $E:= E_j(\varepsilon,h)$
                         which is the image of $R(\varepsilon,h)$ by the diffeomorphism $f_j$ (\ref{new window}).
                    Let $P$ be still one of $P_j^{(\infty)}$. By the proposition
                         (\ref{diffeo}), in the regime $ h \ll \varepsilon $, the differential of
                         $\lambda= P(\xi,\varepsilon;h) $ in
                         $\xi=0$ is a determinant of size $\mathcal{O}(\varepsilon)$:
                          $$\mid det( \frac{\partial P}{\partial \xi} (0,\varepsilon;h)) \mid
                          = M \varepsilon, $$ with $M>0$.
                    By writing the Taylor expansion of $\xi =P^{-1}= g(\lambda, \varepsilon;
                         h)$ in $\lambda_0= P(0,\varepsilon;h)$, we have:
                         $$|\xi| \leq \frac{1}{M \varepsilon} |\lambda - \lambda_0| + \mathcal{O}(|\lambda - \lambda_0|^2) .$$
                    Hence, if $\lambda \in R(\varepsilon,h)$,
                    then $ |\lambda - \lambda_0| \leq \mathcal{O}(h^\delta) $ and so
                         $$|\xi| \leq \frac{1}{M \varepsilon}  \mathcal{O}(h^\delta)  +
                          ( \mathcal{O}(h^\delta))^2
                          \leq \mathcal{O}( \frac{h^\delta}{\varepsilon})  .$$
                    Consequently, from the formula (\ref{hk}), we have that $E:=  E_j(\varepsilon,h)$ is contained in a ball of radius $\mathcal{O}( \frac{h^\delta}{\varepsilon}) $.

 \subsubsection{How is the lattice of quasi-eigenvalues and the lattice of spectrum?}  \label{vi phan lambda}

    For each $j=1, \ldots, L$, from the remark \ref{bien doi sym prin} we can express
                        $P_j^{(\infty)}$ in the form
                        $$\lambda= P_j^{(\infty)}(\xi,\varepsilon;h)=
                          i \varepsilon F_0+ (a_j+i \varepsilon b_j)\cdot \xi+  \mathcal O(\xi ^2)
                                             + \mathcal O(\varepsilon ^2)+ \mathcal O(h) $$
    and thus
                          $$\lambda_1= Re (\lambda) = a_j \cdot \xi +\mathcal O(\xi ^2)
                                             + \mathcal O(\varepsilon ^2)+ \mathcal O(h), $$
                          $$\lambda_2=Im (\lambda)= i \varepsilon F_0 + \varepsilon b_j \cdot \xi + \mathcal O(\xi ^2)
                                             + \mathcal O(\varepsilon ^2)+ \mathcal O(h).$$
    Note that we work in the regime $ h \ll \varepsilon $ and that
                          $$ \frac{ \partial \lambda_1}{ \partial \xi}|_{\xi=0} = a_j + \mathcal O(\varepsilon ^2)+ \mathcal O(h) \sim a_j $$
                          and $$ \frac{ \partial \lambda_2}{ \partial \xi}|_{\xi=0} = \varepsilon  b_j + \mathcal O(\varepsilon ^2)+ \mathcal O(h) \sim b_j \varepsilon .$$
    Thus the variations of the image with respect to the reference value are
                             $$|\Delta \lambda_1| = |a_j | |\Delta \xi| + \mathcal O(|\Delta \xi| ^2) $$
    and
                              $$|\Delta \lambda_2| = \varepsilon |b_j | | \Delta \xi| + \mathcal O(|\Delta \xi|^2).$$
    Hence, if $|\Delta \xi| \sim C h$ ($ C >0 $), then
                          $$|\Delta \lambda_1| \sim  |a_j| \cdot Ch + \mathcal O(h ^2) \sim  C_1h, $$ with $C_1>0$
    and $$|\Delta \lambda_2|  \sim \varepsilon |b_j | \cdot  Ch+ \mathcal O(h ^2)
                          \sim  C_2 \varepsilon h (1+ \mathcal O(\frac{h}{\varepsilon})) \sim  C_2 \varepsilon h,$$ with $C_2>0$.
    Note that $\xi=h(k-\frac{k_j}{4})-\frac{S_j}{2 \pi}, k \in \mathbb Z^2$,
                         then $|\Delta \xi|= h |\Delta k|$ and we can assert that the spectrum of $P_\varepsilon$ in
    a good rectangle $R(\varepsilon, h)$ of the form (\ref{cua so}) is the union of $L$ deformed lattices, with a horizontal spacing $h$ and vertical spacing $\varepsilon h$.

    Of course, the lattices $\Sigma_j(\varepsilon,h)$ are all described the same way. \\
    As a corollary, we have:
                          \begin{rema}
                                The cardinal of such a spectral network
                                in $R(\varepsilon, h)$ is
                                 $\mathcal{O} (\frac{h^\delta . \varepsilon h^\delta}{h . \varepsilon h})=
                                  \mathcal{O}(h^{2(\delta -1)}) $ which converges to
                                  $\infty$ when $h\rightarrow 0$.
                                  This means that the asymptotic expansion is applied to
                                   many eigenvalues of $P_\varepsilon$.

                                   Moreover, a recent work of M.Hitrik-J.Sj\"{o}strand allows us to calculate the cardinal of eigenvalues of $P_\varepsilon$ in the rectangle $R(\varepsilon, h)$.
                          \end{rema}

  \subsection{Birkhoff normal form}  \label{sec FNB}

    \subsubsection{Motivation}

            In  this section, we will discuss the procedure of Birkhoff normal form of a perturbed pseudo-differential  operator $P_\varepsilon$ which depends on small positive parameters
            $h, \varepsilon $ around a Diophantine torus
            $\Lambda$ and treat it explicitly in a particular case (but important for our work).
            For the Birkhoff normal form, we can consult
            \cite{Vu-Ngoc09}, \cite{Charles08}, \cite{Ali85}, \cite{Popov00}.

            We assume that $\Lambda$ is equal to the section $\{\xi=0\}$ in $ T^* \mathbb T^n$ and
            that $ P_\varepsilon$ is microlocally defined near
            $\{\xi=0 \} \in T^* \mathbb T^n$, with $ h $-Weyl (total) symbol $P= P(x,\xi,
            \varepsilon,h)$ which is holomorphic in $(x,\xi)$ near a complex neighborhood of $\xi=0 \in T^* \mathbb T^n$ and $C^\infty$ in $(h, \varepsilon)$ near $0$.

            In the article \cite{Hitrik07} (section 3) one realized the Birkhoff normal form of $P_\varepsilon$ whose $h$-principal symbol is of the form (\ref{symbol micro})
                $$ p_\varepsilon(x,\xi)=p(\xi) +i \varepsilon q(x,\xi)+ \mathcal O (\varepsilon ^2)$$
            and the principal symbol of the obtained normal form is of the form
            (like \ref{prin normal}))
              $$ P_0^{(\infty)}=P_0^{(\infty)} (\xi, \varepsilon)=
               p(\xi)+ i \varepsilon \langle q \rangle (\xi) +  \mathcal O(\varepsilon ^2).$$
            Our work requires treat a more specific case when the principal symbol of $ P_\varepsilon$ already does not depend in $x$ and with no term $\mathcal O (\varepsilon ^2)$:
                $$p_\varepsilon = P_0(\xi, \varepsilon) = p(\xi)+ i \varepsilon  q (\xi).$$

           In the above formula for $ P_0^{(\infty)}$, can we remove the term $\mathcal O (\varepsilon
            ^2)$ and is $P_0^{(\infty)}$ equal to $P_0$ ?
           This is an important result that we want. However, this is not obvious: the proof in \cite{Hitrik07} is not enough to explain it because one used transformations depending also on $\varepsilon$.   \\
           We will prove this result here, by providing a normalization of the total symbol $P(x,\xi, \varepsilon,h)$ in all three variables $(\xi,\varepsilon,h)$.
           The approach that we propose is different from  \cite{Hitrik07}.

    \subsubsection{Formal series and operators}  \label{forme normale de Bir}

         Let $\mathcal{E}= Hol(\mathbb T_x^n) [[\xi, \varepsilon,h]]$ denote the space of formal series in $(\xi,\varepsilon,h)$ with holomorphic coefficients in $x \in \mathbb T^n$,
                    \begin{displaymath}
                        \mathcal{E}=
                         \big \{ a^{(\mathcal{E})}= \sum_{k,m,l=0}^\infty a_{k,m,l}(x) \xi^ k \varepsilon ^m h^l
                        \quad  \textrm{such that }
                         a_{k,m,l}(x)  \textrm{are analytic in $x$ }\big \}.
                     \end{displaymath}
          There is a correspondence between an $h$-Weyl pseudo-differential operator and an element of $\mathcal{E}$:
          if we denote $A$ $h$-Weyl pseudo-differential operator, $a:=\sigma_w(A)$ its Weyl total symbol and $\sigma^{(\mathcal{E})}(A) \in  \mathcal{E} $ the formal Taylor series of $a$ in $(\xi, \varepsilon,h)$, then this correspondence is given by a map, denoted by $\sigma^{(\mathcal{E})}$ from the algebra of (Weyl) pseudo-differential operators $\Psi$ to the space of formal series $\mathcal{E}$:
                        \begin{eqnarray}
                                    \sigma^{(\mathcal{E})}:  \Psi  & \rightarrow &  \mathcal{E}
                                                                                          \nonumber \\
                                                           A    & \mapsto  &   \sigma^{(\mathcal{E})}(A),
                         \end{eqnarray}
                 \begin{displaymath}
                        \sigma^{(\mathcal{E})}(A)= \sum_{ |\alpha |,j,l =0}^\infty
                                \frac{1}{( |\alpha |+ j+l)!} \big ( \partial_\xi^\alpha \partial_\varepsilon^j \partial_h^l
                                 a(x, \xi, \varepsilon, h) \big |_{\xi= \varepsilon=h= 0}   \big )
                                  \xi^\alpha  \varepsilon ^j h^l .
                 \end{displaymath}
          The Moyal formula (see \cite{Moyal49}, \cite{Fedosov96}, \cite{Weyl50})
          for the composition of two operators from the Weyl semi-classical calculation say that if
          $a:=\sigma_w(A), b:=\sigma_w(B)$, then $A\circ B$ is still a $h$-pseudo-differential whose Weyl symbol satisfies
                     \begin{eqnarray} \label{formule de Moyal}
                             \sigma_w(A \circ B) & = &  (a  \sharp^w b)(x, \xi, \varepsilon,h )
                     \nonumber \\
                             & = &  e^{ih[D_\eta D_x- D_y D_\xi]} a(x, \xi, \varepsilon,h )b(x, \xi, \varepsilon,h )
                              \big |_{y=x, \eta= \xi}
                       \nonumber \\
                             & \sim &  \sum_{\alpha,\beta} \frac{h^{|\alpha|+ |\beta|} (-1)^{|\alpha|} }
                             { (2i)^{|\alpha|+ |\beta|} \alpha ! \beta !  }
                             ( \partial_x^\alpha \partial_\xi^ \beta a(x, \xi, \varepsilon,h ) )
                             ( \partial_\xi^\alpha \partial_x^ \beta b(x, \xi, \varepsilon,h ) )
                       \\
                             &= & a(x, \xi, \varepsilon,h ) b(x, \xi, \varepsilon,h ) +
                             \frac{h}{2i} \{ a(x, \xi, \varepsilon,h ) , b(x, \xi, \varepsilon,h ) \} + \cdots
                        \nonumber
                     \end{eqnarray}
       On the other hand, by the Borel theorem, any formal series $a_\mathcal{E} \in \mathcal{E} $ can be seen as Taylor series of a smooth function $a= a(x, \xi, \varepsilon,h )$
       (which is not unique) and we can associate to it a $ h $-pseudo-differential operator by asking
       $A= \textrm{Op}_h^w(a)$. We have then $a^{(\mathcal{E})} = \sigma^{(\mathcal{E})}(A)$.
       In this way and from Moyal formula (\ref{formule de Moyal}), we can define a product on $\mathcal{E}$, denoted by $\star$: let $a^{(\mathcal{E})}, b^{(\mathcal{E})} \in  \mathcal{E} $, then
                        $$a^{(\mathcal{E})}\star b^{(\mathcal{E})}= \sigma^{(\mathcal{E})}(A \circ B) $$
       is the Taylor series of $a  \sharp^w b $. Thus $\mathcal{E}$ becomes an algebra with this product.
       The associative bracket thus is well defined and is called by the Moyal bracket.
       $$ [a^{(\mathcal{E})} , b^{(\mathcal{E})} ]
                        = a^{(\mathcal{E})} \star b^{(\mathcal{E})} - b^{(\mathcal{E})}\star a^{(\mathcal{E})} .$$
       Consequently, let $a^{(\mathcal{E})}\in  \mathcal{E} $, we can define on $\mathcal{E}$ the adjoint operator:
                     $$\textrm{ad}_{a^{(\mathcal{E})}}:=  [a^{(\mathcal{E})}, \cdot ] .$$
       For any $p^{(\mathcal{E})} \in \mathcal{E}$, the formal series
                      \begin{displaymath}
                            e^{ \textrm{ad}_{a^{(\mathcal{E})}} } p^{(\mathcal{E})}=
                            \sum_{k=0}^\infty \frac{1}{k!} (\textrm{ad}_{^{(\mathcal{E})} })^k p^{(\mathcal{E})}
                      \end{displaymath}
       is well defined in $\mathcal{E}$  because it contain only a finite number of terms of fixed degree. Thus the exponential operator $\textrm{exp}( \textrm{ad}_{a^{(\mathcal{E})} }) $ is well defined on $\mathcal{E}$.

       Moreover, for two $h$-pseudo-differential operators $A$ et $P$, we have
                        \begin{displaymath}
                            \sigma^{(\mathcal{E})}([A,P]) = [ \sigma^{(\mathcal{E})}(A),\sigma^{(\mathcal{E})}(P)]
                         \end{displaymath}
       and $\sigma^{(\mathcal{E})}$ becomes an (associative) algebra morphism.

       Similarly, we have also
                     \begin{displaymath}
                            \textrm{exp} (\textrm{ad}_{\sigma^{(\mathcal{E})}(A)} ) (\sigma^{(\mathcal{E})}(P))
                            = \sigma^{(\mathcal{E})} (   \textrm{exp} (\textrm{ad}_ A )  P ),
                      \end{displaymath}
       where the exponential operator $ \textrm{exp} (\textrm{ad}_ A ) $ is defined below.

     \subsubsection{Action by conjugation}

            Let $A$ be a bounded operator on a Hilbert space $H$.
            We define in $\mathcal{B}(H)$ the exponential operator $e^A \in
            \mathcal{B}(H)$ by the absolutely convergent series
            \begin{equation}   \label{dn exp}
                    e^A= \sum_{k \geq 0} \frac{1}{k!} A^k.
            \end{equation}
            Next, we associate with $A$ a bounded operator
            $\textrm{ad}_A:= [A, \cdot] \in \mathcal{B}( \mathcal{B}(H))$ and by (\ref{dn exp}) the operator $\textrm{exp} (\textrm{ad}_A)$ (or $e^{\textrm{ad}_A}$)
            is well defined in $\mathcal{B}(\mathcal{B}(H)) $. In addition, we have the following result:
            \begin{lemm}
                Let $A$ et $P$ be two bounded operators of  $\mathcal{B}(H)$, we have
                \begin{equation}
                        e^A P e^{-A}= \textrm{exp} (\textrm{ad}_A) P  .
                \end{equation}
            \end{lemm}

            \begin{proof}
                For $t \in \mathbb R$, let $f(t)= e^{t A} P e^{-tA}$ and $g(t)=\textrm{exp} (\textrm{ad}_{tA}) P $
                which are analytic functions from $ \mathbb R$ to $\mathcal{B}(H)$.
                 We will calculate its derivatives.
                 First note that $\frac{d}{dt} (e^{tA})= A e^{tA}= e^{tA} A $ and so
                    $$\frac{d}{dt} (e^{tA}) P= A e^{tA}P=   e^{tA} A P ,$$
                 we have
                    $$f'(t)= A e^{tA} P e^{-tA}-  e^{tA} P e^{-tA} A = A f(t)- f(t)A= [A, f(t)]$$
                 and $$g'(t)= \frac{d}{dt} (e^{ t \cdot \textrm{ad}_{A}} P ) =
                     \textrm{ad}_A \circ e^{ t \cdot \textrm{ad}_{A}} P= \textrm{ad}_A g(t)= [A,g(t)]  .$$
                 Then $f(t), g(t)$ satisfy the same linear differential equation of first order
                    $$u'= [A,u].$$
                 But the initial value in $t=0$, $f(0)= g(0)= P$ and consequently we have $f(t)= g(t)$ for all $t \in \mathbb R$. Especially for $t=1$, we well have $ e^A P e^{-A}= \textrm{exp} (\textrm{ad}_A) P$.
            \end{proof}

        As an application, we have:

    \subsubsection{Idea of the construction of Birkhoff normal form}
            The main idea of this construction is to find a pseudo-differential operator $A$ such that
            the associated Fourier integral operator $U(h):=e^{\frac{i}{h} A } $ reduce the initial operator $ P_\varepsilon$ to its conjugate operator
             \begin{equation}
              e^{\frac{i}{h}  A} P_\varepsilon  e^{- \frac{i}{h} A }
             = e^ {\frac{i}{h} \textrm{ad}_{A} } ( P_\varepsilon)
             : = \widetilde {P_\varepsilon }
             \end{equation}
             the formal series of whose total symbol $\sigma^{(\mathcal{E})}( \widetilde {P_\varepsilon }) $
             in $ \mathcal{E}$ does not depend on $x$.
             Noticing that $\sigma^{(\mathcal{E})}( \widetilde {P_\varepsilon })
              = \textrm{exp} (\textrm{ad}_{\sigma^{(\mathcal{E})}(A)} ) (\sigma^{(\mathcal{E})}(P_\varepsilon ))$
             as in the previous section, the work is concentrated to seek $\sigma_{\mathcal{E}}$ as a series of homogenous terms in $(\xi, \varepsilon,h)$.
             In search of this series, the Diophantine condition is essential, see (\ref{serie convergente}).

        \begin{rema}
            \begin{enumerate}
                            \item The above operator $\widetilde {P_\varepsilon }$ is still a pseudo-differential by the Egorov theorem
                             (see for example \cite{Egorov69},\cite{Vu-Ngoc06}, \cite{Hor}).
                             In the special case when $\widehat{A}= h \widehat{B}$, then the operator $U(h):=e^{\frac{i}{h} \widehat{A} } =e^{ i \widehat{B} } $ is really a pseudo-differential operator and $\widetilde {P_\varepsilon }$  is simply the composition of pseudo-differential operators
                            \item The conjugation at the operator level is replaced by the adjoint action on the space of formal series.
            \end{enumerate}
        \end{rema}

     \subsubsection{Construction of the Birkhoff normal form}

        In this part, by convention, for a $ h $-pseudo-differential operator $P$,
        we identify it with its total (Weyl) symbol $P= P(x,\xi,\varepsilon,h)$ and its formal series $\sigma_{\mathcal{E}}(P)$.

        We use a particular order for $(\xi, \varepsilon,h)$ by counting the power in $\xi$ plus twice the power in $\varepsilon $ and $h$.
        The associated filtration is denoted by the symbol $\mathcal O(j)$. \\
        Let us denote also $\mathcal{ D}(j)$ the subspace of $\mathcal{E}$  of homogenous polynomials of degree $j$
        with respect to $(\xi, \varepsilon,h)$ in this order.
        Thus, we have:
                    $$\mathcal{ D}(j) = Vect\{ \xi^ k \varepsilon ^m h^l  | k+ 2(m+l) = j \}
                     \otimes Hol( \mathbb T_x^n) $$
        and
                    $$\mathcal O(j)= \bigoplus_{n \geq j} \mathcal{ D}(j).$$
        We have of course $\mathcal{ D}(j) \subset \mathcal{ O}(j)$ and $\mathcal{ O}(j+1) \subset \mathcal{O}(j)$. As usual, we allow the notation $A= B + \mathcal O(j) $ to say that $A-B \in \mathcal O(j) $.

        If $K_j= K_j(x, \xi, \varepsilon, h) \in \mathcal{ O}(j)$, then it is obvious that the Poisson bracket
             satisfies
            $\{K_j, K_l\} = \mathcal{ O}(j+l-1)$. \\
         For the Moyal bracket $i [K_j, K_l]$, from (\ref{formule de Moyal}), it can be computed as a series in
            $(\frac{h}{i} \frac{\partial}{ \partial \xi} ,\frac{\partial}{ \partial x} )$ and is well of order $j+l+1$
         because every time we lose a degree in  $\xi$ we win also a degree in $h$.
         Moreover, we have
                \begin{eqnarray}  \label{orde}
                         i [K_j, K_l]
                         &  =  &  h \{K_j, K_l\} + h  \mathcal{ O}(j+l)
                                            \\
                         &  =  &   h \mathcal{ O}(j+l-1)  = \mathcal{ O}(j+l+1).
                \end{eqnarray}
         Consequently, we have also
            \begin{eqnarray} \label{orde general}
                [\mathcal{ O}(j_n), [ \ldots   ,[\mathcal{ O}(j_2), \mathcal{ O}(j_1) ] \ldots ]
                &  =  &  h^{n-1} \mathcal{ O} \big(j_1 + \cdots+ j_n-( n-1) \big )
                                       \nonumber      \\
                &  =  &  \mathcal{ O}(j_1 + \cdots+ j_n+ n-1)
                                                    \\
                 \textrm{and}  \quad  [ \ldots [\mathcal{ O}(j_1), \mathcal{ O}(j_2)], \ldots ], \mathcal{ O}(j_n) ]
                 &  =  & h^{n-1} \mathcal{ O} \big(j_1 + \cdots+ j_n-( n-1) \big )
                                              \nonumber       \\
                 &  =  &   \mathcal{ O}(j_1 + \cdots+ j_n+ n-1).
            \end{eqnarray}

     \begin{theo} \label{Theorem FNB}
          Suppose that $P= P(x,\xi, \varepsilon,h)$ is an  analytic Weyl $h$-pseudo-differential operator on $\mathbb T^n$
          (microlocally defined close to $\xi=0$)
          with principal $h$-symbol $P_0= p(\xi)+ i \varepsilon q(\xi)$
          such that $p(\xi)= \langle a, \xi \rangle + \mathcal{ O}(\xi^ 2)$ and $a$ is Diophantine as in the definition \ref{diop}.
         Then for any integer $N \geq 1$, there exists a function $G^{(N)}= \sum_{j=2}^N G_j$ ($G^{(1)}=0$) where
          $G_j=G_j(x,\xi, \varepsilon,h) \in \mathcal{ D}(j-2) $ (for $j \geq 2$)
          is analytic in $x$, homogenous in $(\xi, \varepsilon,h)$ such that
          \begin{equation} \label{normal Bikh}
                \textrm{exp} \big(i \textrm{ad}_{G^{(N)}} \big) P= P_{0} + h P_1^{(N)}+ h R_{N-1},
          \end{equation}
          where $P_1^{(N)}= P_1^{(N)}(\xi, \varepsilon,h) \in \mathcal{E} $ is independent of $x$
          and $ R_{N-1}= \mathcal{ O}(N-1)$.
    \end{theo}

      \begin{proof}
            We can write $P$ in the form
                     $$P= P(x,\xi, \varepsilon,h)= P_{0}+ h P_1 , $$  with $P_1= P_1(x,\xi,\varepsilon,h)$ holomorphic in $(x,\xi)$ and $C^ \infty$ in $(h, \varepsilon)$ close to $0$.
            We will show the property (\ref{normal Bikh}) by induction on $N$.

            For $N= 1$, we take $G^{(1)}= P_1^{(1)} = 0$, $R_0= P_1 (x,\xi, \varepsilon,h) $ and the property (\ref{normal Bikh}) is valid.

            Assume that it is valid for $N$ with a found function $G^{(N)}$.
            We now seek a function $G_{N+1}  \in \mathcal{ D}(N) $ such that $G^{(N+1)}= G^{(N)}+
            G_{N+1}$ satisfies the equation (\ref{normal Bikh}). By developing the exponential and using
            (\ref{orde}) with the attention that $P=  \mathcal O(1)$ and $ G^{(N)} \in \mathcal O(0)$, we can write
            \begin{equation}  \label{exp}
                    \textrm{exp} \big(i \textrm{ad}_{G^{(N+1)}} \big) P=
                     \textrm{exp} \big(i \textrm{ad}_{G^{(N)}} \big) P + i  \textrm{ad}_{G_{N+1}} P + h \mathcal O(N).
            \end{equation}

            Indeed, denoting $A_j:= i  \textrm{ad}_{G_j} $ and $A^{(N)}=  \sum_{j=2}^N A_j$, we have the expansion
            \begin{eqnarray}
                   \textrm{exp} \big(i \textrm{ad}_{G^{(N+1)}} \big) P
                   &  =  &   \textrm{exp} \big( A^{(N)}+ A_{N+1}  \big) P
                                                                          \nonumber \\
                   & = & \sum_{k \geq 0 } \frac{1}{k!} \big(    A^{(N)}+ A_{N+1}  \big)^k P
                                                                           \nonumber \\
                  &  = &  \sum_{k \geq 0 } \frac{1}{k!} \sum_{l=0}^k (A^{(N)})^{k-l} \ast (A_{N+1})^l   P
                                                                            \nonumber \\
                  &  = &  \sum_{k \geq 0 } \frac{1}{k!}  (A^{(N)})^{k} P
                      + \sum_{k \geq 1 }  \underbrace{ \frac{1}{k!}
                                 \sum_{l=1}^k (A^{(N)})^{k-l} \ast (A_{N+1})^l   P }_{B_k}
                                                                            \nonumber \\
                  & = & \textrm{exp} \big( A^{(N)} \big) P + \sum_{k \geq 1 } {B_k}.
            \end{eqnarray}
            In the above formulas, we used the symbol $"\ast" $ which means that $A^m \ast B^n$ is the sum of all compositions containing $m$ times the operator $A$ and $n$ times the operator $B$.

            Particularly, for $k=1$ we have
                $$B_1= A_{N+1} P= i  \textrm{ad}_{G_{N+1}}P= i  \textrm{ad}_{G_{(N+1)}} P .$$
            For all $k \geq 2 $, by using the formula (\ref{orde general}) with remarks
             $P= \mathcal O(1), G^{(N)}=\mathcal O(0), G_{(N+1)}=\mathcal O(N-1) $, we obtain that
             $$B_k = \mathcal O ( N+k ) = h \mathcal O ( N+k -2) $$
             because for $l=1, \ldots, k$ all the terms
             $A^{(N)})^{k-l} \ast (A_{N+1})^l   P = \mathcal O \big( (k-l) \times 0 + l (N-1)+ 1 + k \big)
             =  \mathcal O \big(  l (N-1)+ 1 + k \big)  = \mathcal O ( h ) $ and the inequality $  l (N-1)+ 1 + k \geq N+ k$ is always true.
             Hence the formula (\ref{exp}) is shown.

             From (\ref{exp}), the induction hypothesis (\ref{normal Bikh}) and the formulas (\ref{orde}) one can write
             $$P= \langle a, \xi \rangle + \mathcal{ O}(\xi^ 2)+ i \varepsilon q(\xi) + \mathcal{ O}(h)
             = \langle a, \xi \rangle + \mathcal{ O}(2) ,$$
             we therefore have:
              \begin{eqnarray}
                   \textrm{exp} \big(i \textrm{ad}_{G^{(N+1)}} \big) P
                   &  =  &   P_{0} + h P_1^{(N)}+ h R_{N-1} + i  \textrm{ad}_{G_{N+1}} P + h \mathcal O(N)
                                                                          \nonumber \\
                   & = & P_{0} + h P_1^{(N)}+ h R_{N-1}+ i [G_{N+1},\langle a, \xi \rangle + \mathcal{ O}(2) ]
                                                                           \nonumber \\
                  &  = & P_{0} + h P_1^{(N)}+ h R_{N-1}+ i [G_{N+1},\langle a, \xi \rangle ]
                       + i [G_{N+1}, \mathcal{ O}(2) ]
                                                                            \nonumber \\
                  &  = & P_{0} + h P_1^{(N)}+ h R_{N-1}+ h \{ G_{N+1},\langle a, \xi \rangle \} + h \mathcal{ O}(N-1+ 1)
                        \nonumber \\
                  &   &   + h \mathcal{ O}(N-1+ 2-1)
                                                                             \nonumber \\
                  & = & P_{0}+ h \{ G_{N+1},\langle a, \xi \rangle \}+ h R_{N-1}+ h P_1^{(N)}+ h \mathcal{ O}(N)
            \end{eqnarray}
            Then, the equation for $G_{N+1}$ becomes
                    $$ P_{0}+ h \{ G_{N+1},\langle a, \xi \rangle \}+ h R_{N-1}+ h P_1^{(N)}+ h \mathcal{ O}(N)
                    = P_{0} + h P_1^{(N+1)}+ h \mathcal{ O}(N) $$
            and it is equivalent to the following cohomological equation
            \begin{equation}
                     \{ G_{N+1},\langle a, \xi \rangle \}+  R_{N} = K_{N}+  \mathcal{ O}(N),
            \end{equation}
            where $K_{N}: =  P_1^{(N+1)}- P_1^{(N)} $ should not depend on $ x $.
            Now, this equation is well solvable. \\
            Indeed: in the above equation, as the rest is of order $N$, we can replace $R_{N-1}$ by its homogenous part of order $N-1$, denoted by $\overline{R}_{N-1}$ and we will solve the equation:
            \begin{equation}   \label{pt cohomo}
                     \{ G_{N+1},\langle a, \xi \rangle \}+ \overline{R}_{N-1} = K_{N}.
            \end{equation}

            Develop $G_{N+1}$ and $\overline{R}_{N-1}$ in Fourier series of $x \in \mathbb T^n$
            \begin{eqnarray}
                    G_{N+1} &  =  & \sum_{k \in \mathbb Z^n} \widehat{ G}_{N+1}(k) e^{i k\cdot x}
                        \nonumber \\
                  \overline{R}_{N-1}  & = & \sum_{k \in \mathbb Z^n} \widehat{ \overline{R}}_{N-1} (k) e^{i k\cdot x},
                        \nonumber
            \end{eqnarray}
            where $ \widehat{ G}_{N+1}(k), \widehat{ \overline{R}}_{N-1}(k) $ are polynomials in $\mathbb R [\xi, \varepsilon, h]$.

            Note that the bracket $\{ G_{N+1},\langle a, \xi \rangle \}= - (a \cdot \partial x )   G_{N+1}$, we can write the equation (\ref{pt cohomo}) in the form
            $$
                    - \sum_{k \in \mathbb Z^n} i (a \cdot k) \widehat{ G}_{N+1}(k) e^{i k\cdot x}
                    + \sum_{k \in \mathbb Z^n} \widehat{ \overline{R}}_{N-1} (k) e^{i k\cdot x}=  K_{N}
                $$
            or
                $$
                    \widehat{ \overline{R}}_{N-1} (0)
                    +  \sum_{k \in \mathbb Z^n  \backslash \{ 0 \}}
                     \big ( \widehat{ \overline{R}}_{N-1} (k) - i (a \cdot k) \widehat{ G}_{N+1}(k) \big )
                        e^{i k\cdot x}
                    =  K_{N} .
                $$
            This equation is solved by posing
                    $$ K_{N}= \widehat{ \overline{R}}_{N-1} (0) $$
            (this is also equal to $x$-average $\langle \overline{R}_{N-1} \rangle$ )
            and for $k \in \mathbb Z^n \backslash \{ 0 \}$,
                $$\widehat{ G}_{N+1}(k)= -i \frac{\widehat{ \overline{R}}_{N-1} (k) }{(a \cdot k)},$$
            (here $(a \cdot k) \neq 0$ by the Diophantine condition on $a$).

            In addition, by the Diophantine condition on $a$ (see (\ref{dn alpha-d dioph}) ), there exist two constants $C_0, N_0 > 0$ such that for all $k \neq 0$ we have the estimate:
            \begin{equation} \label{serie convergente}
               \mid \widehat{ G}_{N+1}(k) \mid
                = \frac{ \mid \widehat{ \overline{R}}_{N-1} (k) \mid }{\mid (a \cdot k) \mid}
                \leq C_0 \mid k \mid ^{N_0} \mid \widehat{ \overline{R}}_{N-1} (k) \mid
            \end{equation}
            that ensures convergence and analyticity of $ G_{N+1}$ in $x$ because $\overline{R}_{N-1}$ is.
      \end{proof}

      \begin{rema}

           \begin{enumerate}
                \item In the above theorem, by taking $N$ converge to infinity and by posing $[A]:= G_1+ G_2+
                    \cdots$, then$[A]$ is the desired formal series discussed in the last section and the Birkhoff normal form of $P= P(x,\xi, \varepsilon,h)$ is the limit in $\mathcal{E} $ of
                    $P_{0} + h P_1^{(N)}$ as $N \rightarrow  \infty$. On the other hand, there exits a $C^
                    \infty$ function, denoted often by $P^{(\infty)}$ which admits this limit as its asymptotic expansion.
                 \item  We see an important thing that in the case of theorem, the first term (or yet the $h$- principal term) along the procedure of Birkhoff normal form of $P$ is always
                            $$P_0= p(\xi) + i \varepsilon q(\xi).$$
                \end{enumerate}

       \end{rema}

       \begin{prop}  \label{D.A inverse}
            Let $\widehat{P}= \widehat{P}(\xi; X)$ a complex-valued smooth function of $\xi$
            near $0 \in \mathbb R^2$ and $X$ near $0 \in \mathbb R^n$.
            Assume that $\widehat{P}$ admits an asymptotic expansion in $X$ near $0$ of the form
                $$\widehat{P}(\xi; X) \sim  \sum_\alpha C_\alpha(\xi) X^\alpha$$
            with $C_\alpha(\xi)$ are smooth functions and $C_0(\xi):=\widehat{P}_0(\xi) $ is local diffeomorphism near $\xi=0$.

            Then, for $\mid X \mid $ small enough, $\widehat{P}$ is also a local diffeomorphism near $\xi=0$ and its inverse admits an asymptotic expansion in $X$ near $0 $ whose the first term is
            $(\widehat{P}_0)^{-1}$.
         \end{prop}

          \begin{proof}
            One can write $\widehat{P}$ in the form
            $$\widehat{P}(\xi; X)= \widehat{P}_0(\xi)+ \mathcal O(X). $$
            The determinant
                    $$\mid det( \frac{\partial \widehat{P}}{\partial \xi} (0)) \mid
                  =  \mid det( \frac{\partial \widehat{P}_0}{\partial \xi} (0))  + \mathcal O (\mid X \mid ) \mid $$
            is nonzero for $ \mid X \mid$ small enough and it ensures that $\widehat{P}$ is a local diffeomorphism near $\xi=0$.

            Then, by induction, we can show that $ \widehat {P} ^ {-1} $ admits an asymptotic expansion in
            $X$ near $0 \in \mathbb R^n$.
         \end{proof}

\subsection{Operator $P_\varepsilon= P+ i \varepsilon Q$, the case $ \{p,q\}= 0 $}  \label{sec op comm}

    In this section, we will work on a particular case of the operator $ P_ \varepsilon $ considered in the previous section when the principal symbols $p, q$ commute.
    We now assume that $P_\varepsilon$ if of the form
                \begin{equation}  \label{opérateur}
                        P_\varepsilon= P+ i \varepsilon Q,
                \end{equation}
    with $P,Q$ two $h$-pseudo-differential operators and $P= P_{\varepsilon =0}$ is selfadjoint
                 ($Q$ is not necessarily selfadjoint).

    Suppose that $p,q$ are associated principal symbols of $P, Q$. Note that
                 $p$ is real-valued and we will assume also that $q$ is real-valued.
    Then the principal symbol of $P_\varepsilon$ is
                        $$p_\varepsilon=p+i \varepsilon q. $$
    We assume further that $p,q$ commute i.e. $ \{p,q\}= H_p(q)= 0 $ with respect to the Poisson bracket on $T^*M$
    and that $dp,dq$ are linearly independent almost everywhere.
    \begin{rema}
                    On the operator level, in this case, $P, Q$ are not necessarily in involution, however their commutant is power of order $2$ of $h$,
                              $$[P,Q]= \mathcal O(h^2).$$
    \end{rema}

    \subsubsection{Asymptotic spectrum of $P_\varepsilon= P+ i \varepsilon Q$, the case $ \{p,q\}= 0 $}

            We know that by the commutativity, $ q $ is invariant under the flow of $ p $, the function $\langle q \rangle _T$ (for all $ T> 0 $) in (\ref{t-average}) is hence still $q$
            and by the action-angle theorem \ref{A-A}, $q$ is constant on any invariant torus $\Lambda_a$, the average of $q$ on
                 $\Lambda_a$ (definition in (\ref{moyenne de q})) is still $q$.
            Consequently, we can replace $\langle q \rangle _T$ and $\langle q \rangle $ in all definitions, assumptions and assertions of the last section by $q$.

            Particulary, for each $a \in J$, the compact interval of $Q_\infty(a)$ defined in (\ref{Q vo cung}) becomes a single point:
                            \begin{equation}
                                    Q_\infty(a)=  \{  q _{ | \Lambda_a}  \}
                            \end{equation}
            and locations of the spectrum of $P_\varepsilon$ given in (\ref{loca
                 spectre}), ( \ref{loca. spectre 2}) become:

            \begin{eqnarray}  \label{new loca spectre}
                            Im(\sigma(P_\varepsilon) \cap \{z \in \mathbb C: |Re z| \leq \delta \})
                             \subset &
                                    \varepsilon \big [ \inf_{p^{-1}(0)}
                                    q  -o(1),  \sup_{p^{-1}(0)}
                                    q  + o(1)  \big ]
                             \nonumber \\
                              \subset &
                            \varepsilon  \big [ \inf
                            \bigcup_{a \in J}
                              q _{ | \Lambda_a} -o(1),
                              \sup
                              \bigcup_{a \in J}
                               q _{ | \Lambda_a}  + o(1) \big ],
            \end{eqnarray}
            as $\varepsilon, h, \delta \rightarrow 0$.

            In action-angle coordinates $(x,\xi)$ near $\Lambda_a$ (see the theorem \ref{A-A}) such that
                  $\Lambda_a \simeq \{ \xi=0\}$, we have $p= p(\xi)$ and the formula (\ref{moyenne2}) becomes
                  $\langle q \rangle (\xi)= q(\xi)$. Then microlocally, the principal symbol becomes
                  \begin{equation}
                        p_\varepsilon =p(\xi) +i \varepsilon q(\xi).
                  \end{equation}

        As an important application of section \ref{forme normale de Bir} (the theorem \ref{Theorem FNB}),
        the microlocal construction of Birkhoff quantum normal form of $P_\varepsilon$ in neighborhood of a Diophantine torus does not change
        the principal symbol. I.e in the coordinates $(x,\xi)$ near the section $\xi=0$ in $T^* \mathbb T^2$ such that $P_\varepsilon$ has the normal form,
        its $h$-principal symbol is still of the form $p(\xi) +i \varepsilon q(\xi)$.

        So in this case, concerning the theorem (\ref{theorem quasi-spectre}), the eigenvalues of $P_\varepsilon$ in the rectangle $R(\varepsilon,h)$ (\ref{cua so}) modulo $\mathcal O(h^\infty)$
         are given by the asymptotic expansion of a smooth function $P_j^{(\infty)}$, $j=1, \ldots, L $ in $(\varepsilon,h)$ and in
                   $$\xi = h(k-\frac{k_j}{4})-\frac{S_j}{2 \pi}, k \in \mathbb Z^2 $$
        whose the first term (the principal symbol (\ref{prin normal})
        in the case of the theorem (\ref{theorem quasi-spectre})) is
        \begin{equation}
                           p_{j,0}(\xi,\varepsilon)= p_j(\xi)+ i \varepsilon q_j (\xi),
        \end{equation}
        where $p_j, q_j$ are expressions of $p,q$ in action-angle variables near $\Lambda_j$ of $p,q$. \\
        In reduced form, we can write
                   \begin{eqnarray} \label{new spectre}
                        \sigma(P_\varepsilon)  \cap R(\varepsilon,h) \ni \lambda
                        & = & P_j^{(\infty)} \Big ( h(k-\frac{k_j}{4} )-\frac{S_j}{2 \pi}, \varepsilon;h \Big )
                        + \mathcal O(h^\infty)
                        \nonumber \\
                        & = &
                            p_j\Big (h(k-\frac{k_j}{4} )-\frac{S_j}{2 \pi} \Big )+
                             i \varepsilon q_j \Big (h(k-\frac{k_j}{4})-\frac{S_j}{2 \pi} \Big )
                        \nonumber \\
                               & & + \quad \mathcal O(h), \quad k \in \mathbb Z^2,
                      \end{eqnarray}
        uniformly for $h, \varepsilon$ small.

    \begin{rema}
                        We just give the asymptotic expansion of eigenvalues of
                        $P_\varepsilon$ in a good rectangle in the neighborhood of $0 \in \mathbb C$.

                        However, if we assume the same assumptions on the energy space $p^{-1}(E) \cap
                        T^*M$ ($E \in \mathbb R$) as on $p^{-1}(0) \cap T^*M$
                        (note that if this assumption is satisfied for $E_0$, then it is also satisfied for $E$ in a neighborhood of $E_0$)
                        and by introducing the same definition of set of good values, we can build the same result
                        for the eigenvalues of $P_\varepsilon$ in any good rectangle of center of form $E+ i \varepsilon F_0$.
    \end{rema}

  \subsubsection{The detailed spectral formula}

        In this paragraph, we consider the operator $P_\varepsilon$  as in the preceding paragraph and moreover we will assume all the same assumptions
        on the energy space $p^{-1}(E) \cap T^*M$ as on $p^{-1}(0) \cap T^*M$ and introduce the definition of good values (depending on $E$), similarly to the definition (\ref{dn bonnes valeurs}).
        Here we take $E$ in a bounded interval of $\mathbb R$ because we want uniform estimates with respect to $E$.

        As we said in the previous remark and with the help of the action-angle coordinates, we will explicitly give the asymptotic expansion of eigenvalues
        of $P_\varepsilon$ in an arbitrary good rectangle of size $\mathcal{O}(h^\delta) \times \mathcal{O}(\varepsilon h^\delta)$ with a good center $E+ i \varepsilon F$.
        Moreover, it is interesting that one can construct a such expansion whose principal symbol is globally well defined for all good rectangle near a regular value of $ (p, \varepsilon q) $. \\
        Indeed, for simplicity, we assume that the momentum map $\Phi:=(p,q)$ is proper and has connected fibre.
        \textbf{In this case $L=1$}.

        Denote by $U_r$ the set of regular values of $\Phi=(p,q) $ and let a point $c \in U_r $. \\
        We recall that by the action-angle theorem \ref{A-A}, we have action-angle coordinates in a neighborhood of torus $\Lambda_c:= \Phi^{-1}(c)$ in $M$: there exists $r>0$,
        a neighborhood $\Omega:=\Phi^{-1}(B(c,r)) $ of $\Lambda_c$, an small open $D  \subset \mathbb  R^2$ of center $0$,
        a symplectomorphism $\kappa : \Omega \rightarrow \mathbb T^2 \times D$ and a diffeomorphism $\varphi: D \rightarrow \varphi(D)= B(c,r)$
        such that: $\kappa(\Lambda_c)= \{\xi= 0 \}$, $\Phi \circ \kappa ^{-1}(x, \xi)= \varphi (\xi)$, for all $ x \in \mathbb T^2 , \xi \in D $ and $\varphi(0)=c$.

        We introduce the function
                             \begin{equation}
                                     \chi :  \mathbb R^2 \ni u= (u_1,u_2) \mapsto \chi_u = (u_1, \varepsilon u_2)
                                                     \cong u_1+i \varepsilon u_2
                             \end{equation}
       and denote $$B(\chi_u,r, \varepsilon ):=  \chi(B(u,r)) $$
       for a certain ball $B(u,r)$ ($r>0$), $U_r(\varepsilon):=  \chi(U_r)$.

       For any point $\chi_a \in B(\chi_c,r, \varepsilon )$ such that $F:= a_2$ is a good value, $E:=a_1$ and $p^{-1}(E) \cap T^*M$ satisfies the same assumptions as $p^{-1}(0) \cap T^*M$ in the section \ref{hypothese}.
       We will construct the asymptotic expansion of eigenvalues of $P_\varepsilon $ in a "good rectangle" of "good center" $\chi_a $ (which is just a translation of rectangle $R(\varepsilon,h)$ in (\ref{cua so})):

                     \begin{equation}   \label{new cua so}
                                 R(\chi_a,\varepsilon,h)= \chi_a + R(0, \varepsilon,h),
                     \end{equation}
                     where
                     \begin{equation}
                        R(0, \varepsilon,h)
                                    = \Big[-\frac{h^\delta}{\mathcal{O}(1)},\frac{h^\delta}{\mathcal{O}(1)} \Big]
                                    +i \varepsilon \Big [ -\frac{h^\delta}{\mathcal{O}(1)}, + \frac{h^\delta}{\mathcal{O}(1)} \Big ]
                        = R(\varepsilon,h)-i \varepsilon a_2.
                     \end{equation}
        Let $\Lambda_1= \Phi^{-1}(a)$, this is an invariant torus of type $(\alpha,d)$-Diophantine,
        defined in (\ref{dn alpha-d dioph}) and suppose that its action-angle coordinates are $\xi_a$, i.e. $\{ \xi= \xi_a \} = \kappa (\Lambda_1)$ in $T^* \mathbb T^2$ or $\xi_a= \varphi^{-1}(a) \in D$.

        Then, let $\widetilde{P}_\varepsilon := P_\varepsilon - \chi_a$ to reduce the spectrum of $P_\varepsilon$ near $\chi_a$ to the spectrum of $\widetilde{P}_\varepsilon$ near $0$ by noting that:
            \begin{equation} \label{spec translation}
                \sigma(P_\varepsilon)= \sigma(\widetilde{P}_\varepsilon) +\chi_a .
            \end{equation}

        The principal symbol of $\widetilde{P}_\varepsilon$ is $\widetilde{p}+i  \varepsilon \widetilde{q}$, with $\widetilde{p}:= p-a_1, \widetilde{q}:= q-a_2$.
        Note that we have still an integrable system $( \widetilde{p},\widetilde{q})$ and if we let
        $\widetilde{\xi}= \xi-\xi_a$, then $\widetilde{\xi}$ is the new action variable for this system in which $\Lambda_a \simeq \{\widetilde{\xi}=0 \}$,
        as the standard case of theorem \ref{theorem quasi-spectre}. \\
        The principal symbol of $\widetilde{P}_\varepsilon$ is microlocally reduced to
                $$ \varphi_1(\xi) + i \varepsilon \varphi_2(\xi)-\chi_a=
                \varphi_1(\xi_a + \widetilde{\xi} ) + i \varepsilon \varphi_2(\xi_a + \widetilde{\xi})-\chi_a  .$$
        Applying the theorem (\ref{theorem quasi-spectre}) in the case of last section for $\widetilde{P}_\varepsilon$ and from the formula (\ref{spec translation}),
        we have: all the eigenvalues of $P_\varepsilon$ in the good rectangle $R(\chi_a,\varepsilon,h)$, defined by (\ref{new cua so}) modulo $\mathcal O(h^\infty)$
        are given by the asymptotic expansion of a smooth function $P(\xi, \varepsilon;h)$ in $(\varepsilon,h)$ and in $\xi$ in a neighborhood of $\xi_a$
        such that in reduced form (it's the same as (\ref{new spectre})),
        \begin{eqnarray} \label{spectre TQ}
                        \sigma(P_\varepsilon) \cap R(\chi_a,\varepsilon,h) \ni \lambda
                        & = & P \Big (\xi_a+ h(k-\frac{k_1}{4} )-\frac{S_1}{2 \pi}, \varepsilon;h \Big )
                        + \mathcal O(h^\infty)
                        \nonumber \\
                        & = &
                            \varphi_1\Big (\xi_a+ h(k-\frac{k_1}{4} )-\frac{S_1}{2 \pi} \Big )+
                             i \varepsilon \varphi_2 \Big ( \xi_a+ h(k-\frac{k_1}{4})-\frac{S_1}{2 \pi} \Big )
                        \nonumber \\
                               & & + \mathcal O(h), \quad k \in \mathbb Z^2,
        \end{eqnarray}
        uniformly for $h, \varepsilon $ small, where $S_1 \in \mathbb R^2$ is the action and $k_1 \in \mathbb Z^2$ is the Maslov index of fundamental cycles of $\Lambda_1$.

        With the below remark (\ref{action}), there exists a function $\tau_c \in \mathbb R^2$, locally constant in $c \in U_r$ (depending on the choice of local action-angle coordinates near $c \in U$) such that $ \frac{S_1}{2 \pi}= \xi_a+ \tau_c $.
        So the formula for $\lambda $ becomes:
                      \begin{equation}
                            \lambda =
                            \varphi_1\Big (- \tau_c+ h(k-\frac{k_1}{4} )\Big )+
                             i \varepsilon \varphi_2 \Big ( -\tau_c+ h(k-\frac{k_1}{4})\Big )
                          +\mathcal O(h).
                      \end{equation}
        There is a bijective correspondence between $\lambda \in \sigma(P_\varepsilon) \cap R(\chi_a,\varepsilon,h) $
        and $hk$ in a part of $h \mathbb Z^2 $ (by the proposition \ref{diffeo}).
        Moreover, as in (\ref{hk}) and (\ref{axa spectre}), there exists a smooth local diffeomorphism
                       $f= f(\lambda, \varepsilon;h)$
        which sends $R(\chi_a,\varepsilon,h)$ on its image, denoted by $E(a, \varepsilon,h)$ which is close to $\frac{S_1}{2 \pi}$ such that it sends
                            $\sigma(P_\varepsilon) \cap R(\chi_a,\varepsilon,h) $
                            on $h \mathbb Z^2 $ modulo $\mathcal O(h^\infty)$:
                        $$f= f(\lambda, \varepsilon;h)= \tau_c + h \frac{k_1}{4}+ P^{-1}(\lambda).$$
        Let $\widetilde{f}= f \circ \chi$,
                        $$\widetilde{f}=  \tau_c + h \frac{k_1}{4}+ P^{-1} \circ \chi . $$

    \begin{rema}
          Let $\widehat{P}:= \chi^{-1} \circ P$. Because $P$ admits an asymptotic expansion in $(\xi,
          \varepsilon,h)$, so it is obvious that $\widehat{P}$ admits an asymptotic expansion in
          $(\xi, \varepsilon,\frac{h}{\varepsilon})$ (here $h \ll \varepsilon $):
                \begin{eqnarray}
                        \widehat{P}(\xi, \varepsilon,h) & = &
                            \sum_{\alpha,\beta,\gamma} C_{\alpha\beta\gamma}
                             \xi^\alpha \varepsilon^\beta (\frac{h}{\varepsilon})^\gamma
                       \nonumber \\
                            & = & \widehat{P}_0(\xi) + \mathcal O(\varepsilon, \frac{h}{\varepsilon}),
                \end{eqnarray}
          with $$\widehat{P}_0(\xi)=  \varphi_1(\xi)+ i  \varphi_2(\xi) $$ is a local diffeomorphism.

          Moreover, by looking at the Birkhoff normal form in section \ref{sec FNB}, we can rewrite it as form:
                $$\widehat{P}(\xi, \varepsilon,h) =
                 \widehat{P}_0(\xi) + \mathcal O(\frac{h}{\varepsilon}).$$
          According to the proposition (\ref{D.A inverse}), $(\widehat{P})^{-1}= P^{-1} \circ \chi $ also admits an asymptotic expansion in $(\varepsilon,\frac{h}{\varepsilon})$  whose first term is
                $$(\widehat{P}_0)^{-1}= ( \varphi )^{-1} . $$
    \end{rema}

    Consequently, $\widetilde{f}$ admits an asymptotic expansion in $(\varepsilon, \frac{h}{\varepsilon})$ and
    it can moreover be written as
                 \begin{equation}  \label{DA f still}
                     \widetilde{f}= \widetilde{f}_0 + \mathcal O(\frac{h}{\varepsilon})
                \end{equation}
    whose first term is
            \begin{equation}  \label{pre terme f still}
                 \widetilde{f}_0= \tau_c+ ( \varphi)^{-1}.
            \end{equation}
    We have an important remark that the first term $\widetilde{f}_0$ is well defined globally on $B(c,r)$ in the sense that it does not depends on selected good rectangle $R(\chi_a,\varepsilon,h)$.

        In summary, for any regular value $c \in U_r$, there is a small domain $B(\chi_c,r, \varepsilon)= \chi (B(c,r))$ and for any good value $a \in B(c,r)$ (which is outside a set of small measure), we have a good rectangle $R(\chi_ a,\varepsilon,h)$ of good center $\chi_a$ and a smooth local diffeomorphism $f$ which sends
        $R(\chi_ a,\varepsilon,h)$ on its image, denoted by $E(a,\varepsilon,h)$ of the form:
                        \begin{eqnarray}  \label{new hk}
                                             f :  R(\chi_a,\varepsilon,h) & \rightarrow & E(a,\varepsilon,h)
                                              \nonumber \\
                                       \sigma(P_\varepsilon) \cap R(\chi_a,\varepsilon,h) \ni  \lambda & \mapsto &
                                               f(\lambda,\varepsilon; h) \in h \mathbb Z^2 +\mathcal O(h^\infty).
                        \end{eqnarray}
        such that $\widetilde{f}= f \circ \chi$ admits an asymptotic expansion in
                      $(\varepsilon, \frac{h}{\varepsilon})$ of the form (\ref{DA f still})
       with the first term (\ref{pre terme f still}) is a diffeomorphism, globally defined on $B(c,r)$.

                \begin{figure}[!h]
                \begin{center}
                           \includegraphics[width=0.9 \textwidth]{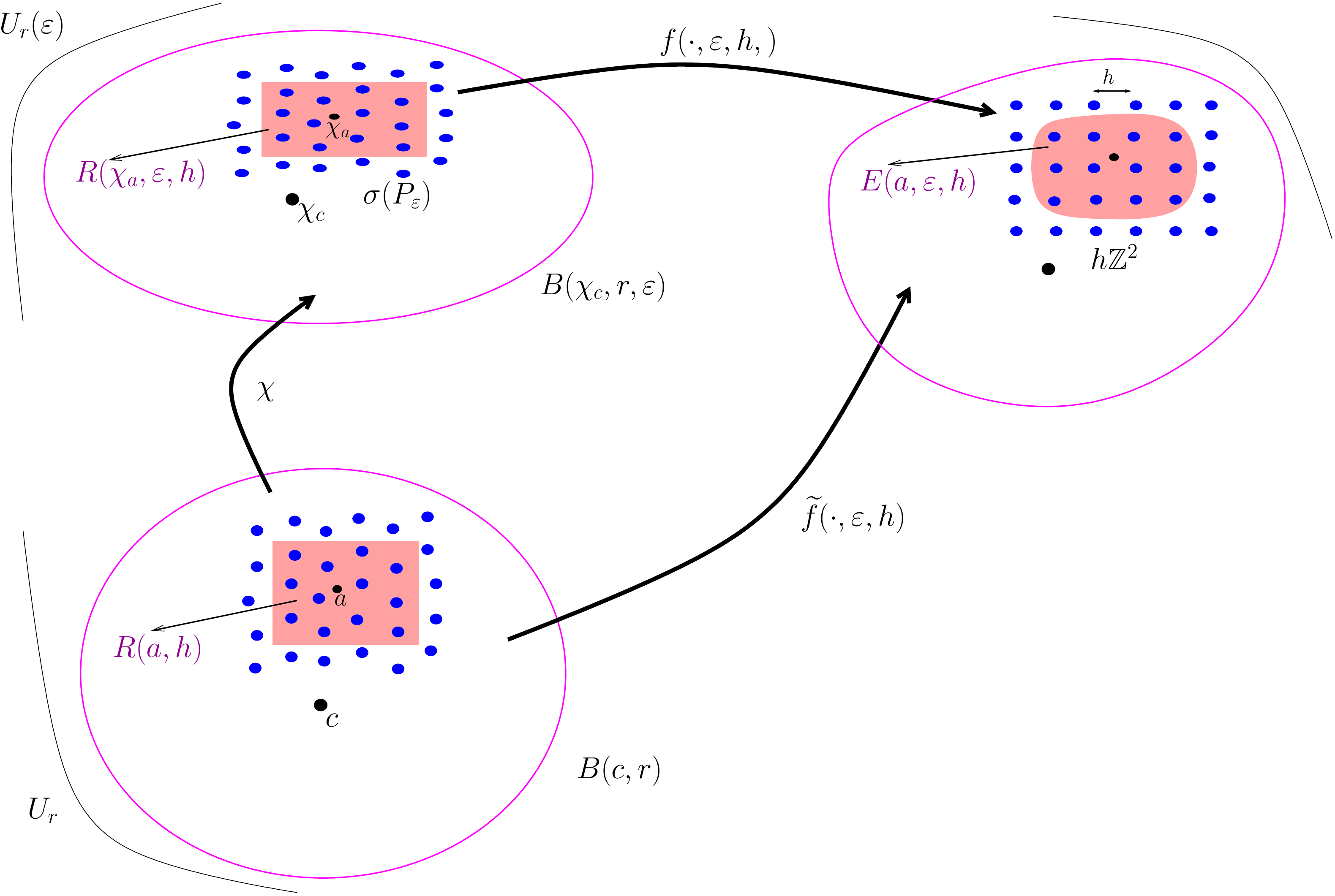} \\
                             \caption{Pseudo-lattice spectrum of $P_\varepsilon$}
                \end{center}
                \end{figure}

     \begin{rema}  \label{symbole f0}
            We don't know if $f$ admits an asymptotic expansion in $(\varepsilon, \frac{h}{\varepsilon})$ but we can express $f$ in the form:
                        $$f(\lambda,\varepsilon,h) =
                        \sum_{\alpha, \beta} C_{\alpha \beta}(\lambda_1, \frac{\lambda_2}{\varepsilon})
                                                      \varepsilon^ \alpha  (\frac{h}{\varepsilon})^\beta  $$
            with
                                  \begin{equation}   \label{f0}
                                      C_{0 0}(\lambda_1, \frac{\lambda_2}{\varepsilon})
                                      = \tau_c +   \varphi^{-1} \circ \chi^{-1}(\lambda):= f_0
                                 \end{equation}
            which is well defined for all  $\lambda \in B(\chi_c,r,\varepsilon)$.\\
            Moreover, we can also write $f=f_0 + \mathcal O(\varepsilon, \frac{h}{\varepsilon}) $.

     \end{rema}

     \begin{rema}
               In terms of the definition \ref{pseu-cart}, we say that couples $\big ( f(\varepsilon;h), R(\chi_a,\varepsilon,h) \big )$ as above form a pseudo-chart $\sigma(P_\varepsilon)$ on the domain  $U_r(\varepsilon)$.
     \end{rema}

     \begin{rema}[Action integral]  \label{action}
        Let $c \in U_r$ a regular value of $\Phi$ and $(x,\xi)$ a set of action-angle coordinates, given by $\kappa$ as before. There is a Liouville $1$- form $\alpha$
        on $\Omega:=\Phi^{-1}(B(c,r)) \subset M $ such that $d \alpha= \omega$.

        Let $\widetilde{\omega}$ the canonical symplectic form and $\widetilde{\alpha}= \sum \xi_i dx_i$ a canonical form on  $T^* \mathbb T^2$: $d\widetilde{\alpha}= \widetilde{\omega} $. \\
       As $\kappa$ is symplectic, we have  $\kappa ^*\omega = \widetilde{\omega}$.
       This is equivalent to $d (\kappa ^* \alpha- \widetilde{\alpha})= 0 $ and there is a $ 1 $-closed form $\beta$ on $T^* \mathbb T^2$ such that
                                  $$\kappa ^* \alpha= \widetilde{\alpha}+  \beta. $$
       For any invariant torus $\Lambda_a \subset \Omega$, let $(\gamma_1,\gamma_2)$ two fundamental cycles on $\Lambda_a$ that are sent on the sides of the torus:
                                   $ \kappa (\Lambda_a)= \{\xi=  \xi_a\}$ by $\kappa$ in $T^* \mathbb T^2$.
        So the action of $(\gamma_1,\gamma_2)$ is $S_1= (S_{1,1}, S_{1,2})$, calculated by: for $j=1,2,$
                                 \begin{eqnarray}
                                          S_{1,j} & = & \int_{\gamma_j}
                                              \alpha=
                                          \int_{\kappa(\gamma_j)}
                                             \kappa^* \alpha
                                          =
                                          \int_{\kappa(\gamma_j)}
                                             (\widetilde{\alpha}+  \beta)
                                    \nonumber \\
                                          & = &  \int_{\{x \in \mathbb T^2: x_j=0 \} }
                                             (\sum \xi_i dx_i+  \beta)
                                          = 2  \pi ( \xi_a + \tau_{c,j}),  \nonumber
                                  \end{eqnarray}
    where  $$\tau_{c,j}:=  \int_{ \{x \in \mathbb T^2: x_j=0 \}} \beta $$
    is a constant, independent of $\Lambda_a \subset \Omega$ (independent of $c \in U_r $) by the closedness of $\beta$. \\
    Then, there exists a function $\tau_c \in \mathbb R^2$, locally constant in $c \in U_r$ such that
        \begin{equation}  \label{pt action}
                                    \frac{S_1}{2 \pi} = \xi_a + \tau_c.
        \end{equation}
    \end{rema}

\subsection{Construction of the monodromy of asymptotic pseudo-lattice}
    The spectrum of the operator $ P_ \varepsilon $ considered in the previous section is a model of a more general lattice that we define and discuss below.

    Let $U$ a subset of $\mathbb R^2$  with compact closure and denote $U(\varepsilon)= \chi(U)$ where $\chi$ is the function defined as in previous section.
    Let $\Sigma(\varepsilon, h)$ (which depends on small $h$ and $\varepsilon$) a discrete set of $U(\varepsilon)$.

     \begin{defi} \label{pseu-cart}
         For $h, \varepsilon$ small enough and in the regime $ h \ll \varepsilon$,
         we say that $(\Sigma(\varepsilon, h), U(\varepsilon))$ is an asymptotic pseudo-lattice if:
         for any small parameter $\alpha >0$, there exists a set of "good values" in $\mathbb R^2$, denoted by $BV$
         whose complement is of small measure in the sense
             $$\mid {}^C BV  \cap I  \mid \leq C \alpha \mid I \mid$$
         for any domain $I \subset \mathbb R^2$ and $C>0$ is a constant.  \\
        For all $c \in U$, there exists a ball $B(c,r) \subset U $ around $c$ ($r>0$) such that for every "good value"
        $a= (a_1,a_2) \in  B(c,r) \cap  BV$,
        there is a good rectangle $R(\chi_a,\varepsilon,h) \subset \chi(B(c,r))$ of good center $\chi_a$:
                        $$R(\chi_a,\varepsilon,h)= \chi (R(a,h))$$
        where $ R(a,h)$ is a rectangle of size $\mathcal{O}(h^\delta) \times \mathcal{O}( h^\delta) $, $0< \delta <1$
        and a smooth local diffeomorphism (in $\chi_a$) $f= f(\cdot; \varepsilon,h)$
        which sends $R(\chi_ a,\varepsilon,h)$ on its image, denoted by $E(a,\varepsilon,h)$ satisfying
        \begin{eqnarray}  \label{semi-cart}
                                             f :  R(\chi_ a,\varepsilon,h) & \rightarrow & E(a,\varepsilon,h)
                                                        \nonumber \\
                                       \Sigma( \varepsilon, h) \cap R(\chi_a,\varepsilon,h) \ni  \lambda & \mapsto &
                                               f(\lambda;\varepsilon, h) \in h \mathbb Z^2 +\mathcal O(h^\infty)
                        \end{eqnarray}
        such that $\widetilde{f}:= f \circ \chi$ admits an asymptotic expansion in $(\varepsilon,\frac{h}{\varepsilon})$ for the $C ^\infty$ topology for the variable $u$ in a neighborhood of $a$ and in the reduced form,                     \begin{equation}  \label{pt khai trien}
                             \widetilde{f}(u)= \widetilde{f}_0(u) + \mathcal O(\varepsilon, \frac{h}{\varepsilon}),
        \end{equation}
        where the first term $\widetilde{f}_0$ is a diffeomorphism, independent of $\alpha$, globally defined on $B(c,r)$ and independent of the chosen good value $a \in B(c,r)$.

        We also say that the family of $(f(\cdot; \varepsilon,h), R(\chi_ a,\varepsilon,h) )$ is a local pseudo-chart on $B(\chi_c, r, \varepsilon):= \chi(B(c,r)$ and that a couple $(f(\cdot; \varepsilon,h), R(\chi_ a,\varepsilon,h) )$ is a micro-chart of $(\Sigma(\varepsilon, h), U(\varepsilon))$. \\
 \end{defi}

                 \begin{figure}[!h]
                 \begin{center}
                           \includegraphics[width=0.9  \textwidth]{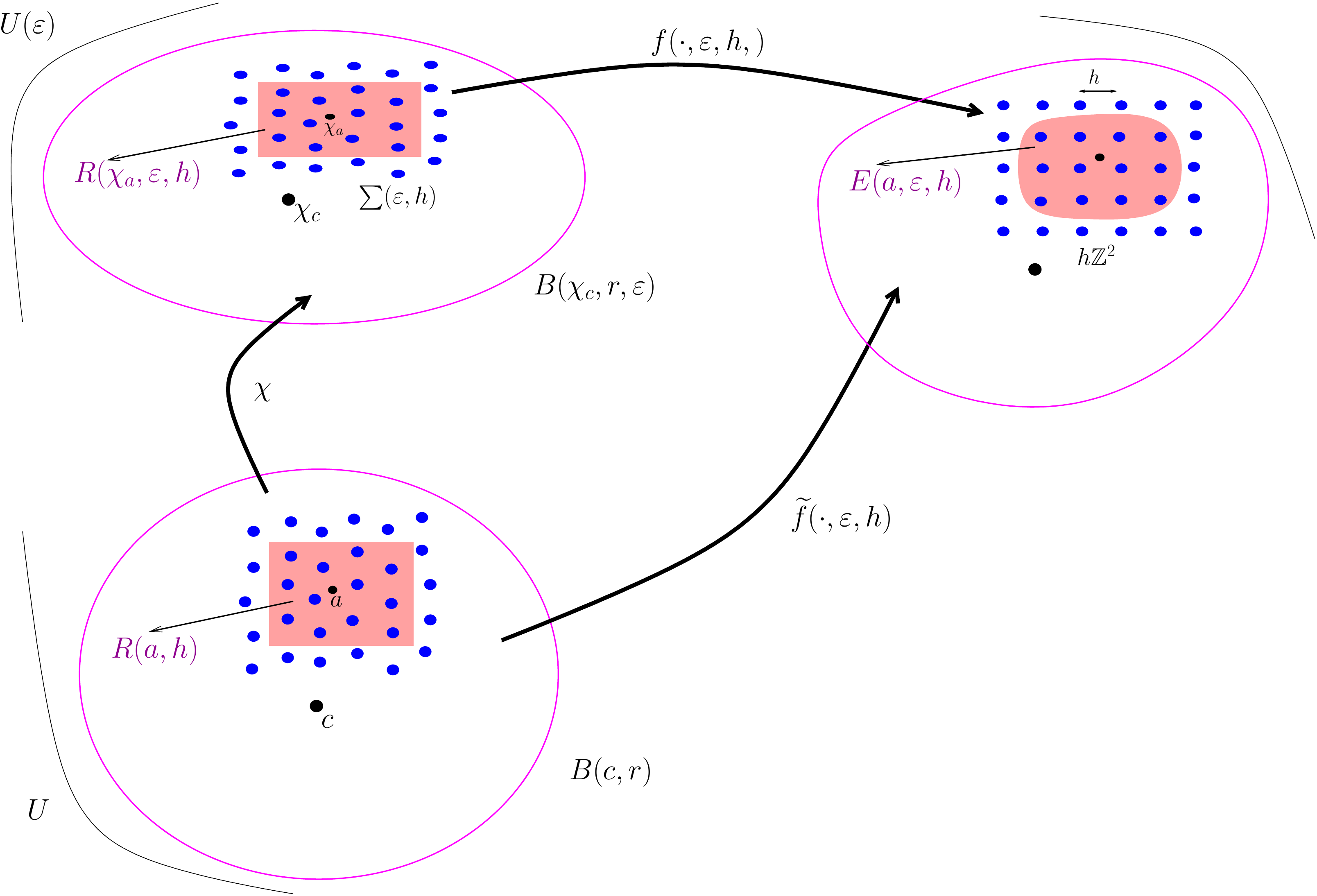} \\
                             \caption{Asymptotic pseudo-lattice}
                 \end{center}
                \end{figure}

    \begin{rema}
          It is clear that the spectrum of an operator $P_\varepsilon = P + i\varepsilon Q$ considered in the previous section is a good example of this definition.
          In this case, $ \widetilde {f} _0 $ is equal to actions coordinate.

          We want to define a combinatorial invariant (spectral monodromy) of $\Sigma( \varepsilon, h)= \sigma(P_\varepsilon)$.
          As we know, in this case $ P, Q $ does not necessarily commute, so it can not have any joint spectrum as the integrable case that we discussed.
          Therefore, it is not clear if one can define the monodromy for the spectrum of  $P_\varepsilon$.

          On the other hand, we are careful that the map $ f $ in (\ref{new hk}) is not an affine chart of
          $U (\varepsilon) $ defined in previous section because it is only defined on a domain depending on $ h $ which will be reduced to a single point when $h \rightarrow 0 $.
          Therefore, we can not apply the construction of the quantum monodromy for an affine asymptotic lattice as the article \cite{Vu-Ngoc99}.
          However, we can successfully build this invariant for the discrete spectrum of $ P_ \varepsilon $ due to the fact that the first term $ \widetilde {f} _0 $ is globally defined on a small ball $ B (c, r) $.

    \end{rema}

     \begin{lemm}  \label{ton tai suite}
        Let $(\Sigma(\varepsilon, h), U(\varepsilon))$ an asymptotic pseudo-lattice as in the definition \ref{pseu-cart} and a point $\chi_a \in  B(\chi_c,r, \varepsilon ) $ with $a$ a good value.
        Then, there is a family $  \lambda(\varepsilon,h)  \in  \Sigma(\varepsilon, h) \cap R(\chi_a,\varepsilon,h)$ such that
                            \begin{eqnarray}
                                  |\lambda_1(\varepsilon,h) -a_1|= \mathcal O (h)
                                                \\
                                 |\lambda_2 (\varepsilon,h) - \varepsilon a_2|= \mathcal O ( \varepsilon .h),
                           \end{eqnarray}
        uniformly for $h, \varepsilon \rightarrow 0$.
     \end{lemm}

     \begin{proof}

        By the proposition \ref{D.A inverse}, $u:= \widetilde{f}^{-1}$ is also a local diffeomorphism in a fixed point $\xi_a$ admitting an asymptotic expansion in $(\varepsilon, \frac{h}{\varepsilon})$
        and by its main part, if $\xi$ is near $\xi_a$ we have:
                         \begin{equation}
                             u(\xi)= \widetilde{f}^{-1} (\xi; \varepsilon, h)
                             = (\widetilde{f}_0)^{-1}(\xi) + \mathcal O(\varepsilon, \frac{h}{\varepsilon}).
                        \end{equation}
       By calculating the differential of $u$ in $\xi_a$ (this is the same as the proof of proposition \ref{D.A
                        inverse}), we can assert that:
       if $\xi^{(1)},\xi ^{(2)}$ near $\xi_a$ such that $ |\xi^{(1)}-\xi^{(2)}|= \mathcal O (h) $, then
                                         $$|u(\xi^{(1)})-u(\xi^{(2)}) |= \mathcal O (h)  $$
       uniformly for $h, \varepsilon \rightarrow 0$.    \\
       Let
                         $$\xi^{(1)}:= \widetilde{f}(a, \varepsilon; h ) \in E(a,\varepsilon,h ). $$
       On the other hand, one can find
                        $$ \xi^{(2)}:= hk= h .k(\varepsilon, h) \in h \mathbb Z^2 \cap E(a,\varepsilon,h )$$
       such that $ |\xi^{(1)}-\xi^{(2)}| \leq h$ by setting $k= k(\varepsilon, h)$ the integer part of
                         $\frac{\xi^{(1)} }{h}$.
       In the end, with the remark $\chi (u(\xi^{(1)}))= \chi_a$ and there is
                         $  \lambda(\varepsilon,h)  \in  \Sigma(\varepsilon, h) \cap R(\chi_a,\varepsilon,h) $
       such that $  \lambda(\varepsilon,h) = \chi(u(\xi^{(2)}))  + \mathcal O (h^\infty)$,
       we get the result of the lemma.
     \end{proof}

     \begin{rema}  \label{k=k(h)}
             The result of the lemma is still valid for the spectrum of an operator in general case of the theorem  (\ref{theorem quasi-spectre}).

             In the previous proof, we can choose $k$ as the integer part of  $\frac{ \widetilde{f}_0(a)}{h}$ and then $k= k(h)$.
     \end{rema}

\subsubsection{Transition function} \label{sec fonc de transition}

    Let $(\Sigma(\varepsilon, h), U(\varepsilon))$ be an asymptotic pseudo-lattice.

    Suppose that $B _\alpha:= B(c,r)$ and $B_\beta:= B(c', r')$ are two small balls in $U$ with nonempty intersection
    $B _{\alpha \beta}:=  B _\alpha \cap B_\beta \neq \emptyset $ such that there are two local pseudo-charts on
    $B _{\alpha}(\varepsilon):= B(\chi_c, r, \varepsilon)$ and $B _{\beta}(\varepsilon)= B(\chi_c', r', \varepsilon)$ of $(\Sigma(\varepsilon, h), U(\varepsilon))$.
    Denote $B _{\alpha \beta}(\varepsilon) = B _{\alpha}(\varepsilon) \cap B _{\beta}(\varepsilon) $. \\
    Because the good values in each $B _{\alpha}$, $B _{\beta}$ are outside the set of small measure $\mathcal O(\alpha)$, then the complement of good values in $B _{\alpha \beta}$
    has still small measure $\mathcal O(\alpha)$.

    Let $a \in B _{\alpha \beta} $ be a good value. Therefore, there is an associated good rectangle
    $ R(\chi_a,\varepsilon,h) \subset  B _{\alpha \beta}(\varepsilon)$ on which we have two micro-charts of $\Sigma(\varepsilon, h)$ in $\chi_a$ as in (\ref{semi-cart}) of the previous definition:
            \begin{eqnarray}
                         f_\alpha (\varepsilon;h): R(\chi_a,\varepsilon,h)\rightarrow E_\alpha(a,\varepsilon,h)
                            \nonumber \\
                        f_\beta (\varepsilon;h): R(\chi_a,\varepsilon,h) \rightarrow E_\beta(a,\varepsilon,h).
            \end{eqnarray}

     \begin{theo}  \label{df trans en bon valeur}
            There exists a unique constant matrix $M_{\alpha \beta} \in GL(2, \mathbb Z)$ such that
                $$d\widetilde{f}_{\alpha, 0}(a)= M_{\alpha \beta} d\widetilde{f}_{\beta, 0}(a)$$
            at all good values $a \in B _{\alpha \beta}$.
     \end{theo}

     \begin{proof}
        In this proof, we keep the same notation as lemma \ref{ton tai suite}.
        By this lemma \ref{ton tai suite} and the remark \ref{k=k(h)}, let
             \begin{equation} \label{pt 1}
                \lambda(\varepsilon,h) =f_\alpha^{-1}( hk(h)) + \mathcal O (h^\infty)
            \end{equation}
        (with $ h k(h) \in h \mathbb Z^2 \cap E_\alpha(a,\varepsilon,h) $) is a family in
                            $ \Sigma(\varepsilon, h) \cap R(\chi_a,\varepsilon,h)$
        such that  $$|\lambda_1(\varepsilon,h) -a_1|= \mathcal O (h) $$ and
                   $$|\lambda_2 (\varepsilon,h) - \varepsilon a_2|= \mathcal O ( \varepsilon .h) ,$$
        uniformly for $h, \varepsilon \rightarrow 0$.

        Let $k_0 \in  \mathbb Z^2$ be arbitrarily given. For $h$ small enough, we can define
                                    \begin{equation} \label{pt hk}
                                             hk'(h):= hk(h)- hk_0 \in h \mathbb Z^2 \cap
                                             E_\alpha(a,\varepsilon,h) .
                                     \end{equation}
        Then there exists a family
                            $\lambda'(\varepsilon,h) \in   \Sigma(\varepsilon, h) \cap R(\chi_a,\varepsilon,h)$
        such that            \begin{equation}  \label{pt 2}
                                hk'(h)= f_\alpha ( \lambda'(\varepsilon,h))+ \mathcal O (h^\infty),
                            \end{equation}
        uniformly for $\varepsilon, h$ small. \\
        We have also
                             \begin{equation} \label{pt 3}
                                hk(h)= f_\alpha ( \lambda(\varepsilon,h))+ \mathcal O (h^\infty),
                            \end{equation}
        uniformly for $\varepsilon, h$ small.
        By replacing (\ref{pt 2}) and (\ref{pt 3}) in (\ref{pt hk}), we have:
                            \begin{equation} \label{pt 4}
                                f_\alpha ( \lambda(\varepsilon,h)) - f_\alpha ( \lambda'(\varepsilon,h))
                                = hk_0 +   \mathcal O (h^\infty),
                            \end{equation}
                            uniformly for $\varepsilon, h$ small. \\
        Note also that $ |hk'(h)- hk(h) | = hk_0 = \mathcal O (h)$, then
                             $ |hk'(h)- \xi^{(1)} |=  \mathcal O (h)$
        and so the family $\lambda'(\varepsilon,h)$ has the same property as that of $\lambda(\varepsilon,h)$.
        That is
                                     $$|\lambda'_1(\varepsilon,h) -a_1|= \mathcal O (h) $$
               and
                                    $$|\lambda'_2 (\varepsilon,h) - \varepsilon a_2|=
                                    \mathcal O ( \varepsilon .h) ,$$
                            uniformly for $h, \varepsilon \rightarrow 0$.

        We recall the function
                             \begin{eqnarray}
                                     \chi :B _{\alpha \beta} \rightarrow B _{\alpha \beta}(\varepsilon)
                                 \nonumber \\
                                        u= (u_1,u_2) \mapsto \chi_u= (u_1, \varepsilon u_2)
                             \end{eqnarray}
       Let $u (\varepsilon,h)= \chi^{-1}(\lambda(\varepsilon,h)) $
                             (i.e. $u_1= \lambda_1, u_2= \frac{\lambda_2}{\varepsilon}$)
       and in the same way $u' (\varepsilon,h)= \chi^{-1}(\lambda'(\varepsilon,h)) $. We have                             \begin{eqnarray}
                                         |u_i (\varepsilon,h) - a_i| = \mathcal O (h),
                                 \nonumber \\
                                        |u'_i (\varepsilon,h) - a_i| = \mathcal O (h),
       \end{eqnarray}
                        for $i=1,2$.
       The equation (\ref{pt 4}) so becomes
                            \begin{equation*}
                                \widetilde{f}_\alpha ( u(\varepsilon,h)) - \widetilde{f}_\alpha ( u'(\varepsilon,h))
                                = hk_0 +   \mathcal O (h^\infty),
                            \end{equation*}
        or
                            \begin{equation} \label{pt 5}
                               \frac { \widetilde{f}_\alpha ( u(\varepsilon,h)) - \widetilde{f}_\alpha ( u'(\varepsilon,h))}
                               {h}
                                = k_0 +   \mathcal O (h^\infty),
                            \end{equation}
        uniformly for $\varepsilon, h$ small. \\
        Because we can express
                            \begin{equation*}
                                 \widetilde{f}_\alpha ( u(\varepsilon,h))= \widetilde{f}_{\alpha, 0} ( u(\varepsilon,h))
                                       + \mathcal O(\varepsilon, \frac{h}{\varepsilon})
                              \end{equation*}
        and by writing the Taylor expansion of $ \widetilde{f}_{\alpha, 0}( u(\varepsilon,h))$ in $a$ with the integral remainder,
        by doing the same work for $  \widetilde{f}_\alpha ( u'(\varepsilon,h)) $ and using that if $R(u,\varepsilon,h)=\mathcal O(\varepsilon, \frac{h}{\varepsilon})$, then
                $$|R( u(\varepsilon,h),\varepsilon,h)-R(u'(\varepsilon,h),\varepsilon,h) | =
                                         \mathcal O (h) \times \mathcal O(\varepsilon, \frac{h}{\varepsilon}),$$
         uniformly for $ h, \varepsilon $ small and $h \ll \varepsilon $
         as $ |u(\varepsilon,h) - u'(\varepsilon,h) |= \mathcal O (h)  $,
         equation (\ref{pt 5}) gives us:
         \begin{equation*}
                                ( d \widetilde{f}_{\alpha, 0})(a)
                               \frac { u(\varepsilon,h) -  u'(\varepsilon,h)}{h}
                                = k_0 +  \mathcal O(\varepsilon, \frac{h}{\varepsilon})
         \end{equation*}
         Consequently, we have
                            \begin{equation} \label{pt 6}
                               \frac { u(\varepsilon,h) -  u'(\varepsilon,h)}{h}
                                = \big ( ( d \widetilde{f}_{\alpha, 0})(a) \big ) ^{-1}( k_0) +
                                  \big ( ( d \widetilde{f}_{\alpha, 0})(a) \big ) ^{-1}
                                  \big(  \mathcal O(\varepsilon, \frac{h}{\varepsilon}) \big ).
                            \end{equation}

         On the other hand, as the norm of differential $ \big ( ( d \widetilde{f}_{\alpha, 0})(a) \big ) ^{-1}$ is a constant independent of $\varepsilon,h$, the equation (\ref{pt 6})  allows us to write
                             \begin{equation} \label{pt 7}
                               \frac { u(\varepsilon,h) -  u'(\varepsilon,h)}{h}
                                = \big ( ( d \widetilde{f}_{\alpha, 0})(a) \big ) ^{-1}( k_0)
                                   +  \mathcal O(\varepsilon, \frac{h}{\varepsilon}),
                            \end{equation}
         uniformly for $ h, \varepsilon $ small and $h \ll \varepsilon $.

         Now we will work with $f_\beta (\varepsilon;h)$.
         Let also $\widetilde{f}_{\beta}= f_\beta \circ \chi$. \\
         Because $ \lambda(\varepsilon,h),  \lambda'(\varepsilon,h) $ is in $\sigma(P_\varepsilon) \cap R(\chi_a,\varepsilon,h)$, then there exists a family, denoted by $k'(\varepsilon, h) \in \mathbb Z^2$ such that \begin{equation*}
            \frac { \widetilde{f}_\beta ( u(\varepsilon,h)) - \widetilde{f}_\beta ( u'(\varepsilon,h))}{h}
                                =   k'(\varepsilon, h) + \mathcal O (h^\infty),
         \end{equation*}
         uniformly for $\varepsilon, h$ small. \\
         In the same way as before, we also get
                                \begin{equation} \label{pt 8}
                                   \frac { u(\varepsilon,h) -  u'(\varepsilon,h)}{h}
                                   = \big ( ( d \widetilde{f}_{\beta, 0})(a) \big ) ^{-1}( k'(\varepsilon, h) )
                                      + \mathcal O(\varepsilon, \frac{h}{\varepsilon}),
                                \end{equation}
          uniformly for $\varepsilon, h$ small and $h \ll \varepsilon $. \\
          Then the equation (\ref{pt 7}) and the equation (\ref{pt 8}) give us
                                \begin{equation*}
                                        \big ( ( d \widetilde{f}_{\alpha, 0})(a) \big ) ^{-1}( k_0)
                                        = \big ( ( d \widetilde{f}_{\beta, 0})(a) \big ) ^{-1}( k'(\varepsilon, h) )
                                          +   \mathcal O(\varepsilon, \frac{h}{\varepsilon}),
                                \end{equation*}
          uniformly for $\varepsilon, h$ small, $h \ll \varepsilon $ and therefore
                                 \begin{equation}
                                        ( d \widetilde{f}_{\beta, 0})(a)
                                        \big ( ( d \widetilde{f}_{\alpha, 0})(a) \big ) ^{-1}( k_0)
                                        =  k'(\varepsilon, h)
                                          +  \mathcal O(\varepsilon, \frac{h}{\varepsilon}),
                                \end{equation}
           uniformly for $\varepsilon, h$ small and $h \ll \varepsilon $. \\
           As the left part of the last equation is a constant, $\mathcal O(\varepsilon, \frac{h}{\varepsilon})$ is small for $\varepsilon, h$ small, $h \ll \varepsilon $ and $k'(\varepsilon, h) \in \mathbb Z^2$,
           it is necessary that $k'(\varepsilon, h) \in \mathbb Z^2$ is a constant $k' \in \mathbb Z^2$
          and that we have
                                  \begin{equation}
                                        ( d \widetilde{f}_{\beta, 0})(a)
                                        \big ( ( d \widetilde{f}_{\alpha, 0})(a) \big ) ^{-1}( k_0)
                                        =  k' \in  \mathbb Z^2
                                \end{equation}
         for all $k_0 \in \mathbb Z^2$ given. \\
         This means that  $( d \widetilde{f}_{\beta, 0})(a)
                                        \big ( ( d \widetilde{f}_{\alpha, 0})(a) \big ) ^{-1} \in GL(2, \mathbb Z) $.
         ~~\\
         On the other hand, $ d \widetilde{f}_{\beta, 0} \circ \big (  d \widetilde{f}_{\alpha, 0} \big ) ^{-1}$  is uniformly continuous on $B _{\alpha \beta}$
         and the measure of complementary of good values in $B _{\alpha \beta}$ is $\mathcal O(\alpha)$. By taking $\alpha$ small enough and
         with the fact that the group $GL(2, \mathbb Z)$ is discrete, the uniform continuity implies that there is a constant matrix $M_{\beta \alpha } \in GL(2, \mathbb Z)$ such that
                                $$ ( d \widetilde{f}_{\beta, 0})(a) =
                                M_{\beta \alpha } ( d \widetilde{f}_{\alpha, 0})(a) ,$$
         for all good values $a \in B _{\alpha \beta}$.
     \end{proof}

     On intersections, the function $f_{\alpha, 0} \circ (f_{\beta, 0})^{-1}$ is well defined, smooth. Thus, outside a set of measure $\mathcal O(\alpha)$, the value of  $ d( f_{\alpha, 0} \circ (f_{\beta, 0})^{-1})$ is a constant matrix $M_{\alpha \beta} \in GL(2, \mathbb Z)$.
     Note that $ d( f_{\alpha, 0} \circ (f_{\beta, 0})^{-1})$ is independent of $\alpha$.
     By taking $\alpha \rightarrow 0$, we get that $ d( f_{\alpha, 0} \circ (f_{\beta, 0})^{-1})$ is equal to the constant $M_{\alpha \beta}$ outside a set of measure $0$.
     By continuity, we must have $ d( f_{\alpha, 0} \circ (f_{\beta, 0})^{-1})= M_{\alpha \beta}$ everywhere.
     Then we have:

     \begin{theo} \label{df transition}
            There exists a unique constant matrix $M_{\alpha \beta} \in GL(2, \mathbb Z)$ such that
                                $$d\widetilde{f}_{\alpha, 0}= M_{\alpha \beta} d\widetilde{f}_{\beta, 0}$$
            on $B _{\alpha \beta} $.
     \end{theo}

     \begin{rema} \label{avec a-a}
            For the spectrum of the operator $P_\varepsilon = P + i\varepsilon Q$ considered in the previous section,
            we can see that the result of the theorem can be found independently from classical results.
            Indeed: \\
            We remember from (\ref{pre terme f still}) that the leading terms are
                             \begin{eqnarray}
                                  \widetilde{f}_{\alpha, 0}= \tau_c+ \varphi_\alpha^{-1},
                                      \nonumber \\
                                  \widetilde{f}_{\beta, 0}= \tau_{c'}+ \varphi_\beta^{-1} .
                             \end{eqnarray}

       On the other hand, as an application of action-angle theorem
       (see the next section and the equation (\ref{ctruc aff})), on nonempty $B _{\alpha \beta} $, there exists an affine map $A _{\alpha \beta} \in GA_{\mathbb R}(2, \mathbb Z) $ of the form
                $$A _{\alpha \beta}(\cdot):= M _{\alpha \beta}^{cl}(\cdot)+ C_{\alpha \beta} ,$$
       with a matrix $M _{\alpha \beta}^{cl} \in GL(2, \mathbb Z), C_{\alpha \beta} \in \mathbb R^2$ such that
        \begin{equation} \varphi_\alpha^{-1} \circ \varphi_\beta  = A _{\alpha \beta}.
        \end{equation}
       Therefore the action coordinates $\xi_a$ and $\xi_a'$ of $\Lambda_a$ in two action-angle coordinates, associated with $\varphi_\alpha, \varphi_\beta$ satisfy
                                $$\xi_a=  A _{\alpha \beta}\xi_a' .$$
       We have also the corresponding integral actions on $\Lambda_a$ satisfying the relation
                            $$ S_1= M _{\alpha \beta}^{cl} S_1'. $$
       Two last equations and the equation (\ref{pt action}) give us
                                $$\tau_c= M _{\alpha \beta}^{cl} \tau_c' - C_{\alpha \beta}.$$
       Consequently, we obtain
                         \begin{equation} \widetilde{f}_{\alpha, 0}= (M _{\alpha \beta}^{cl}) \widetilde{f}_{\beta, 0}.
                         \end{equation}
       It means that we recover the result of theorem with help of action-angle theorem. \\
       Moreover, if we denote $M_{\alpha \beta}^{sp}$ the matrix $M_{\alpha \beta}$ defined by the theorem \ref{df transition}, we have:
                             $$M_{\alpha \beta}^{sp}= (M _{\alpha \beta}^{cl}) . $$

       \end{rema}

         \begin{rema}
             For an asymptotic pseudo-lattice, we could try to show a stronger result: there exists a unique constant matrix $M_{\alpha \beta} \in GL(2, \mathbb Z)$ such that
                                $$ \widetilde{f}_{\alpha, 0}= M_{\alpha \beta} \widetilde{f}_{\beta, 0}$$
             on $B _{\alpha \beta} $.

             By the remark \ref{avec a-a}, we have just seen that this result is true for spectrum of the discussed operator $P_\varepsilon$.
             However, we have not succeeded to show it in the case of a general asymptotic pseudo-lattice.

         \end{rema}

         \subsubsection{Definition of the monodromy of an asymptotic pseudo-lattice}
                Let $(\Sigma(\varepsilon, h), U(\varepsilon))$ be an asymptotic pseudo-lattice as in \ref{pseu-cart}.
                Assume that $U$ is covered by a locally finite covering $ \{ B_{\alpha} \} $
                and that $(\Sigma(\varepsilon, h), U(\varepsilon))$ is covered by associated local pseudo-charts on $B_{\alpha} (\varepsilon):= \chi(B_{\alpha})$:
                              $$ \{ (f(\cdot; \varepsilon,h), R(\chi_ a,\varepsilon,h) )
                                \quad a \in  B_{\alpha}  \quad \textrm{is a good value } \} .$$
                We can see $\{(\widetilde{f}_{\alpha,0}( \varepsilon, h),B_\alpha ) \} $ as the charts of $U$ whose transition functions are in the linear group $GL(2, \mathbb Z)$. \\
                Indeed, on each nonempty intersection
                     $ B_{\alpha} \cap B_{\beta} \neq \emptyset$, the theorem \ref{df transition} states that
                there exists a unique linear map $M_{\alpha \beta} \in GL(2, \mathbb Z)$ satisfying
                    \begin{equation}
                    d(\widetilde{f}_{\alpha, 0}) = M_{\alpha \beta} d(\widetilde{f}_{\beta, 0}) \qquad \textrm{ or} \qquad
                    d \big ( \widetilde{f}_{\alpha, 0} \circ (\widetilde{f}_{\beta, 0})^{-1} \big )= M_{\alpha \beta}.
                    \end{equation}

                 The family $\{M_{\alpha \beta}\} $ defines a $1$-cocycle, denoted by $\mathcal M$, in the \v{C}ech cohomology with value in the linear group $GL(2, \mathbb Z)$.

                 \begin{defi} \label{monodromie pseu}
                            The class $[\mathcal M] \in \check{H}^1(U(\varepsilon),GL(2, \mathbb Z) )$ is called the (linear) monodromy of the asymptotic pseudo-lattice
                            $(\Sigma(\varepsilon, h), U(\varepsilon))$.
                 \end{defi}

                 As with quantum monodromy, we can of course associate this cocycle class with an isomorphism class of bundles (bundle up to isomorphism) on $U(\varepsilon)$ with structure of the group $GL(2, \mathbb Z)$ and fiber $\mathbb Z^2$. The transition functions between two adjacent trivializations of the bundle are $\{M_{\alpha \beta}\} $.

                 Let $\mathcal M$  is some cocycle associated with trivialization of the bundle and let $\gamma(\varepsilon)$ be a closed loop, contained in $U(\varepsilon)$.
                 We define
                                $$\mu(\gamma(\varepsilon))=
                                    M_{1,N} \circ M_{N,N-1} \circ \cdots M_{3,2} \circ  M_{2,1},$$
                 where $M_{i,j}$ denotes the corresponding transition function to each pair of nonempty intersection $(B_i(\varepsilon), B_j(\varepsilon)) $, defined as in the theorem (\ref{df transition}),
                 here $(B_1(\varepsilon), \cdots, B_N(\varepsilon))$ is a numbered finite covering of $\gamma(\varepsilon)$ in $U(\varepsilon)$ with $B_i \cap B_{i+1} \neq \emptyset $.

                 The non triviality of $[\mathcal M ]$ is equivalent to that of a group morphism, denoted also by $\mu$, defined as follows
                 \begin{eqnarray}
                                \mu : \pi_1(U(\varepsilon)) & \rightarrow &  GL(2, \mathbb Z)/ \{ \sim\}
                             \nonumber \\
                                   \gamma (\varepsilon)  & \mapsto &  \mu(\gamma(\varepsilon)) ,
                    \end{eqnarray}
                 where $ \{ \sim\}$ denote the modulo conjugation. \\
                 We call $\mu$ the representation of the monodromy $[\mathcal M]$.

   \subsection{Linear Spectral Monodromy}
        We consider again the operator  $P_\varepsilon= P+ i \varepsilon Q$, in the case $ \{p,q\}= 0 $ that we discussed in section \ref{sec op comm}.

        Let $U$ be a subset of regular values $U_r$ of the map $(p,q)$ with compact closure and denote $U(\varepsilon)= \chi(U)$. \\
        We further assume that the spectrum of $P_\varepsilon$ is discrete in $U(\varepsilon)$. Then it is clear that $(\sigma(P_\varepsilon), U(\varepsilon) )$ is an asymptotic pseudo-lattice, adapted to the definition \ref{pseu-cart}.
        We can therefore define the monodromy of $P_\varepsilon$ as the monodromy of the asymptotic pseudo-lattice $(\sigma(P_\varepsilon), U(\varepsilon) )$ according to the definition  \ref{monodromie pseu}.
        We call it the (linear) spectral monodromy, denoted by
                $[\mathcal M _{sp}] \in \check{H}^1(U(\varepsilon),GL(2, \mathbb Z))$.

        \begin{defi} \label{monodromie spectrale}
               For $\varepsilon, h > 0$ small such that $h \ll \varepsilon \leq h^\delta, 0< \delta< 1$,
               the class $[\mathcal M _{sp}] \in \check{H}^1(U(\varepsilon),GL(2, \mathbb Z) )$ is called the spectral monodromy of the operator $P_\varepsilon$ on the domain $U(\varepsilon)$.
        \end{defi}

   \section{Relationship between the Spectral Monodromy and the Classical Monodromy }

   We will also make the link with classical monodromy that illuminates the existence of linear spectral monodromy.
   First, we recall the classical monodromy that is given by J. Duistermaat in the article \cite{Duis80}.

   \subsection{Classical Monodromy}
        Let $ (M, \omega) $ be a symplectic manifold of dimension $ 2n $.
        \begin{defi}
         A completely integrable system on $M$ is given $n$ real-valued functions $f_1, \dots, f_n$
         in $ C^\infty(M)$ in involution whose differentials are almost everywhere linearly independent.  \\
         In this case, the map $ F=(f_1, \dots, f_n ): M \rightarrow \mathbb R^n $ is  called \textbf{ \textit{momentum map}}.
        \end{defi}

        A point $m \in M $ is regular for the momentum map $F$ if its differential $dF(m)$ has maximal rank $n$.
        For $c \in \mathbb R ^n$, we say that $\Lambda_c $ is a leaf of $F$ if it is a connected component of $F^{-1}(c)$.
        And moreover this leaf is regular if all its points are regular point for $ F $.

        The "action-angle theorem" of Liouville, Mineur and Arnold
          (see \cite{Mineur48}, \cite{Hofer95}, \cite{Audin01}, \cite{Cushman97}) says that if $\Lambda_c $ is compact,
          then there exists local "action-angle coordinates" in a neighborhood of $\Lambda_c $. However, we have maybe no global existence of these action-angle coordinates.
          A obstruction of that global existence is a geometrical invariant, called monodromy, given the first time in 1980 by J. Duistermaat in the article \cite{Duis80}.
          For more on this monodromy, we can also see \cite{Vu-Ngoc01}, \cite{Vu-Ngoc06}.

  \begin{theo}[Action-Angle theorem]  \label{A-A}
    Let $F=( f_1,\dots , f_n )$ be completely integrable system and $\Lambda_c$ be a compact regular sheet of $F$. Then there exists a neighborhood $\Omega$ of $\Lambda_c$ in $M$, a small open disk $D$  with center $c$ in $ \mathbb R^n$ and a symplectomorphism $\Psi: \mathbb T ^n \times D \rightarrow \Omega$ such that:
        \begin{enumerate}
            \item $\Psi (\mathbb T ^n \times \{c \})= \Lambda_c $.
            \item $\Omega$ is saturated, i.e. all leaves that pass through a point of $\Omega$ are tori,
                included in $\Omega$.
            \item We have $$F\circ \Psi(x, \xi)= \varphi(\xi)$$ for all $x=(x_1, \dots, x_n) \in \mathbb T ^n $, all $\xi= ( \xi _1, \dots \xi_ n ) \in D$ and here $\varphi: D \rightarrow \varphi(D)$ in a (local) diffeomorphism with $\varphi(c)= c$. ~~\\
                In particular, the flow of all $f_i$ in $\Omega$ is complete. ~~\\
                On $\Omega$, $( x, \xi)= \Psi^{-1}$ and we say that $x \in \mathbb T ^n $ are (local) angle variables and $\xi \in D \subset \mathbb R^n $ are (local) action variables.
        \end{enumerate}
  \end{theo}

  ~~\\
   Recall that $c \in \mathbb R^n$ is a regular value of $F$ if all points of the fiber $F^{-1}(c)$ are regular points for $F$.
   Let $U_r \subseteq \mathbb R^n$ be the set of regular values of $F$.
   We assume moreover that $ F $ is proper (so that the fibers are compact), $B_r $ and fibers $ F ^ {-1} (c) $ are connected.
   In this case, we can apply the action-angle theorem at each point of the set of regular leaves $F^{-1}(B_r) \subseteq M $.
   Moreover,  there exists in general a integer structure on the space of regular leaves (see \cite{Duis80}, \cite{Vu-Ngoc06}):
   let $U_ \alpha$,$U_\beta$ be any two small enough open sets in $U_r$ with nonempty intersections such that there are action-angle coordinates on $F^{-1}(U_\alpha) \subset M $
   and $F^{-1}(U_\beta) \subset M $. With notation as in the previous theorem, then on $D_\alpha \cap D_\beta \neq \emptyset $, the function $\varphi_{\alpha \beta}:= \varphi_\alpha ^ {-1} \circ \varphi_\beta$ is a affine map: for $\xi \in D_\alpha \cap D_\beta $,
     \begin{equation}     \label{ctruc aff}
      \varphi_{\alpha \beta}(\xi)= A_{\alpha \beta} \cdot \xi= M_{\alpha \beta} \cdot \xi+ C_{\alpha \beta},
     \end{equation}
    where $M_{\alpha \beta}$ is a integer constant matrix of $GL(n, \mathbb Z)$ and $ C_{\alpha \beta} \in \mathbb R^n$ is a constant.

    On the other hand, for all $c \in U_r$, on $\Lambda_c$, as the flow of each $f_i$  is complete , then the joint flow of $F$, denoted by  $\varphi^t_F$ defines a locally free group action $(\mathbb R^n, +)$, $$ \varphi^t_F: \mathbb R^n \times \Lambda_c \rightarrow \Lambda_c  .$$
   We define the stabiliser of torus $\Lambda_c$, denoted by $ S_c$, which is defined independently of the choice of $m$ in $\Lambda_c$:
                                    \begin{equation} \label{stabiliseur}
                                            S_c = \{ \tau \in \mathbb R^n: \varphi^ \tau_F(m)= m \}.
                                    \end{equation}
   ~~\\
   It is know that $S_c$ is a discrete subgroup of rank $n$ of $\mathbb R^n$ (so isomorphic to $\mathbb Z^n$) and the set of all these stabilizers forms a bundle, called the period bundle over $U_r$,  $ \mathscr F: S_c \rightarrow c \in U_r$. It is locally trivial by the action-angle theorem, but can be globally nontrivial. \\
   Moreover, with notation as in action-angle theorem, for all $b \in \varphi(D) \subseteq B_r$ near $c$, a basic of $S_{b}$, denoted by $ (\tau)(b) $ is given by the formula:
        $$(\tau)(b)= (\tau^{(1)}(b), \ldots,\tau^{(n)}(b) )= [ d (\varphi^{-1})(b)]  .$$
   Then, from previous relation and by noting the equation (\ref{ctruc aff}), we obtain that the transition function between two trivialization of the bundle $\mathscr F $ are $ \{ {}^tM_{\alpha \beta}^{-1} \} $ in $GL(n, \mathbb Z)$- the integer linear group.

   On the other hand, by noting that a $n-$tuple $\tau=( \tau_1,\ldots, \tau_n) \in \mathbb R^n$ is such that the Hamiltonian $ \tau_1 f_1 + \cdots + \tau_n  f_n$
   admits on the torus $\Lambda_c$ a periodic flow of period $1$ if and only if $\tau \in S_c$, then we get so an isomorphism between the stabiliser $S_c$ and the homology group $H_1(\Lambda_c, \mathbb Z)$:
                                        $$ S_c \simeq H_1(\Lambda_c, \mathbb Z) .$$
   Then the bundle $\mathscr F $ can be identified with the homology bundle
    $H_1(\Lambda_c, \mathbb Z) \rightarrow c \in B_r$.
    The possible nontriviality of this bundle is called (linear classical) monodromy of completely integrable system $F$.
    The troviality of this monodromy is equivalent to the global existence of action variables on the space of regular sheets.

\subsection{Relationship} \label{relation avec classique}

        We recall that the (linear) classical monodromy is given by J.Duistermaat
                          \cite{Duis80} (presented in the previous section)
        is defined as a bundle $H_1(\Lambda_c, \mathbb Z) \rightarrow c \in U $, associated with a cocycle, denoted by $[\mathcal M _{cl} ]$ in $\check{H}^1(U,GL(2, \mathbb Z) ) $ of transition functions:
                         $$\{ {}^t ( M_{\alpha \beta}^{cl})^{-1}
                         = {}^t (d(\varphi_\alpha^{-1} \circ \varphi_\beta ) )^{-1} \}.$$

        The remark (\ref{avec a-a}) gives us thus the following relationship between two monodromy types.

        \begin{theo} \label{dinh ly relation}
                The linear spectral monodromy is the adjoint of the linear classical monodromy
                                    $$ [\mathcal M _{sp} ]= {}^t[\mathcal M _{cl} ]^{-1} . $$
                In other words,if the corresponding representations of monodromy of
                            $ [\mathcal M _{sp}]$ and $[\mathcal M _{cl}]$ are
                                \begin{eqnarray}
                                    \mu^{sp} : \pi_1(U(\varepsilon)) & \rightarrow &  GL(2, \mathbb Z)/ \{ \sim\}
                                                 \nonumber \\
                                    \mu^{cl} : \pi_1(U) & \rightarrow &  GL(2, \mathbb Z)/ \{ \sim\},
                                \end{eqnarray}
                then $\mu^{sp}= {}^t(\mu^{cl})^{-1} $.
        \end{theo}

        \begin{rema}
                A particular case happens when $[P,Q]=0$, this implies $\{p,q\}= 0$. We can have two spectral monodromy types for $P_\varepsilon$: the affine spectral monodromy, defined in section 2 and the linear spectral monodromy, defined in this section. It is obvious that the linear spectral monodromy is the linear part of the affine spectral monodromy.
                ~~\\
                Note also that in this case, by definition, the affine spectral monodromy is equal to the quantum monodromy and it is known from a result of S. Vu Ngoc \cite{Vu-Ngoc99}
                that the latter monodromy is the adjoint of classical monodromy. That gives once again the result of previous theorem in the integrable quantum case.
        \end{rema}

\bigskip

~\\

{\sc Q. S. Phan}\\
Universit\'e d'Aix-Marseille\\
CPT-CNRS, Luminy, Case 907\\
13288 Marseille, France\\
E-mail: quang-sang.phan@univ-amu.fr\\


\begin{thebibliography} {[10]}
\frenchspacing \baselineskip=12 pt plus 1pt minus 1pt




\bibitem{Ali85} {\sc M. K. Ali}, {\em The quantum normal form and its equivalents}, J. Math. Phys. {\bf 26}
    (1985), no.10, p. 2565-2572.


\bibitem{Arnold67} {\sc V. I. Arnold}, {\em On a characteristic class entering into conditions of
              quantization}, Funkcional. Anal. i Prilo\v zen. {\bf 1} (1967), p. 1-14.



\bibitem {Audin01}  {\sc  M. Audin}, {\em Les syst\`{e}mes hamiltoniens et leur int\'{e}grabilit\'{e}}, Cours
    Specialis\'{e}s  [Specialized  Courses], vol. 8, Soci\'{e}t\'{e} Math\'{e}matique de France, Paris, 2001.



\bibitem {Bambusi99}  {\sc D. Bambusi, S. Graffi and T. Paul}, {\em  Normal forms and quantization
    formula}, Comm. Math. Phys. {\bf 207} (1999), no. 1, p. 173-195.

\bibitem {Havl08}  {\sc  J. Blank, P. Exner and M. Havl{\'{\i}}{\v{c}}ek}, {\em Hilbert space operators
    in quantum physics}, second ed., Theoretical and Mathematical Physics, Springer, New York, 2008.


\bibitem {Bost86} {\sc J. B. Bost}, {\em Tores invariants des systemes dynamiques hamiltoniens}
    (d'apr\`{e}s Kolmogorov, Arnold, Moser, Russmann, Zehnder, Herman,
    Poschel, : : :), Asterisque (1986), no. 133-134, p. 113-157, Seminar Bourbaki, Vol. 1984/85.

\bibitem {Broer91} {\sc H. W. Broer and G. B. Huitema}, {\em  A proof of the isoenergetic KAM-theorem
    from the "ordinary" one}, J. Differential Equations {\bf 90} (1991), no. 1, p. 52-60.

\bibitem {Broer10} {\sc H. Boer}, {\em Do Diophantine vectors form a Cantor bouquet?}, J. Difference Equ.
    Appl. {\bf 16} (2010), no. 5-6, p. 433-434.

\bibitem {Broer07} {\sc H. Broer, R. Cushman, F. Fass\`{o} and F. Takens}, {\em Geometry of KAM tori for
    nearly integrable Hamiltonian systemsh}, Ergodic Theory Dynam. Systems  {\bf 27} (2007), no. 3, p. 725-741.



\bibitem {Burns01} {\sc D. Burns and R. Hind}, {\em Symplectic geometry and the uniqueness of
    Grauert tubes}, Geom. Funct. Anal. {\bf 11} (2001), no. 1, p. 1-10.

\bibitem {Cappel94} {\sc S. E. Cappell, R. Lee and E. Y. Mimmer}, {\em On the
    Maslov index},  Comm. Pure Appl. Math. {\bf 47} (1994), no. 2, p. 121-186.


\bibitem {Charbonnel88} {\sc A. M. Charbonnel}, {\em Comportement semi-classique du spectre
    conjoint  d'op\'{e}rateurs pseudo diff\'{e}rentiels qui commutent},  Asymptotic Anal. {\bf 1} (1988), no. 3, p. 227-261.

\bibitem {Charles08} {\sc L. Charles and S. V{\~u} Ng{\d{o}}c}, {\em Spectral asymptotics
    via the semiclassical Birkhoff normal form}, Duke Math. J. {\bf 143} (2008), no. 3, p. 463-511.

\bibitem {Cushman97} {\sc R. H. Cushman nad  L. M. Bates}, {\em Global aspects of classical integrable
    systems},  Birkhäuser Verlag, Basel, 1997.



\bibitem {Davies95} {\sc E. B. Davies}, {\em Spectral theory and differential operators}, Cambridge
    Studies in Advanced Mathematics, vol. 42, Cambridge University Press, Cambridge, 1995.

\bibitem {Delshams96} {\sc A. Delshams and P. Guti\'{e}rrez }, {\em “Effective stability and KAM
    theory}, J. Differential Equations {\bf 128} (1996), no. 2, p. 415-490.


\bibitem {Dimas99} {\sc M. Dimassi and J. Sj\"{o}strand}, {\em  Spectral asymptotics in the
    semi-classical  limit},  London Mathematical Society Lecture Note Series, vol. 268, Cambridge University Press, Cambridge, 1999.


\bibitem {Duis80} {\sc  J. J. Duitermaat}, {\em  On global action-angle coordinates},  Comm. Pure
    Appl. Math. {\bf 33} (1980), no. 6, p. 687-706.



\bibitem {Eckhardt86}{\sc B. Eckhardt}, {\em Birkhoff-Gustavson normal form in classical and
    quantum mechanics}, J. Phys. A  {\bf  19} (1986), no. 15, p. 2961-2972.


\bibitem {Egorov69} {\sc J. V. Egorov}, {\em The canonical transformations of
    pseudodifferential operators},  Uspehi Mat. Nauk {\bf 24} (1969), no. 5 (149), p. 235-236.

\bibitem {Fedosov96} {\sc B. Fedosov}, {\em Deformation quantization and index theory},  Mathematical
    Topics, vol. 9, Akademie Verlag, Berlin, 1996.



\bibitem {Hitrik06} {\sc M. Hitrik}, {\em Lagrangian tori and spectra for non-selfadjoint operators},
    Séminaire: Équations aux Dérivées Partielles. 2005–2006, Sémin. Équ. Dériv. Partielles, École Polytech., Palaiseau, 2006,
    Based on joint works with J. Sjöstrand and S. V˜u Ngo. c, p. Exp. No. XXIV, 16.


\bibitem {Hitrik04} {\sc M. Hitrik and J. Sj{\"o}strand}, {\em Non-selfadjoint perturbations of
    selfadjoint operators in 2 dimensions. I},  Ann. Henri Poincaré {\bf 5} (2004), no. 1, p. 1-73.


\bibitem {Hitrik05} {\sc M. Hitrik and J. Sj{\"o}strand}, {\em Nonselfadjoint perturbations of selfadjoint
    operators in two dimensions. II. Vanishing averages}, Comm. Partial Differential Equations {\bf 30} (2005), no. 7-9, p.
    1065-1106.


\bibitem {Hitrik08} {\sc M. Hitrik and J. Sj{\"o}strand}, {\em Non-selfadjoint perturbations of selfadjoint
    operators in two dimensions. IIIa. One branching point}, Canad. J. Math. {\bf 60} (2008), no. 3, p. 572-657.


\bibitem {Hitrik07} {\sc M. Hitrik and J. Sj{\"o}strand and S. V{\~u} Ng{\d{o}}c}, {\em  Diophantine tori
    and spectral  asymptotics for nonselfadjoint operators}, Amer. J. Math. {\bf 129} (2007), no. 1, p. 105-182.



\bibitem {Hofer95} {\sc H. Hoffer and E. Zehnder}, {\em Symplectic invariants and Hamiltonian
    dynamics}, The Floer memorial volume, Progr. Math., vol. {\bf 133}, Birkhäuser, Basel, 1995, p. 525-544.



\bibitem {Hor} {\sc L. H\"{o}rmander}, {\em  The analysis of linear partial differential operators.
    I-IV}, Classics in Mathematics, Springer- Verlag, Berlin, 1983-1990, Fourier integral operators, Reprint of the 1994 edition.


\bibitem {Kan06} {\sc S. J. Kan}, {\em  Erratum to the paper: “On rigidity of Grauert tubes over
    homogeneous Riemannian manifolds}, [J. Reine Angew. Math. 577 (2004), 213-233; mr2108219]”, J. Reine Angew. Math. {\bf 596} (2006), p. 235.




\bibitem {Kato95} {\sc T. Kato}, {\em  Perturbation theory for linear operators}, Classics in Mathematics,
    Springer-Verlag, Berlin, 1995, Reprint of the 1980 edition.


\bibitem {Mineur48} {\sc H. Mineur}, {\em  Quelques propri\'{e}t\'{e}s g\'{e}n\'{e}rales des \'{e}quations de la
    m\'{e}canique}, Bull. Astr. (2) {\bf 13} (1948), p. 309-328.


\bibitem {Moyal49} {\sc J. E. Moyal}, {\em  Quantum mechanics as a statistical theory},
    Proc. Cambridge Philos. Soc. {\bf 45} (1949), p. 99-124.

\bibitem {Oliveira09} {\sc C. R. De Oliveira}, {\em  Intermediate spectral
    theory and quantum dynamics}, Progress in Mathematical Physics, vol. 54, Birkhäuser Verlag, Basel,
    2009.

\bibitem {Popov00} {\sc G. Popov}, {\em  Invariant tori, effective stability, and quasimodes with
    exponentially small error terms. I-II. Birkhoff normal forms},  Ann. Henri Poincaré {\bf 1} (2000), no. 2, p. 223-279.


\bibitem {Poschel82} {\sc J. P\"{o}schel}, {\em  Integrability of Hamiltonian systems on Cantor sets},
    Comm. Pure  Appl. Math. {\bf 35} (1982), no. 5, p. 653-696.


\bibitem {RS78} {\sc  M. Reed and B. Simon}, {\em  Methods of modern mathematical physics. IV.
    Analysis of operators},  Academic Press [Harcourt Brace Jovanovich Publishers], New York, 1978.


\bibitem {Robert87} {\sc  D. Robert}, {\em Autour de l’approximation semi-classique},  Progress in
    Mathematics, vol. 68, Birkhäuser Boston Inc., Boston, MA, 1987.


\bibitem {Shubin01} {\sc  M. A. Shubin}, {\em  Pseudodifferential operators and spectral theory},
    second éd., Springer-Verlag, Berlin, 2001, Translated from the 1978 Russian original by Stig I. Andersson.

\bibitem {Sj00.1} {\sc  J.  Sj\"{o}strans}, {\em  Asymptotic distribution of eigenfrequencies  for
    damped wave equations},  Publ. Res. Inst.  Math. Sci. {\bf 36} (2000), no. 5, p. 573-611.


\bibitem {Sj00.2} {\sc  J.  Sj\"{o}strans}, {\em  Asymptotic distribution of eigenfrequencies for damped wave
    equations}, Journées “Équations aux Dérivées Partielles” (La Chapelle sur Erdre, 2000), Univ. Nantes, Nantes, 2000, p. Exp.
    No. XVI, 8.

\bibitem {Sj09} {\sc J.  Sj\"{o}strans }, {\em  Eigenvalue distribution for non-self-adjoint operators with small
    multiplicative  random perturbations}, Ann. Fac. Sci. Toulouse Math. (6) {\bf 18} (2009), no. 4, p. 739-795.

\bibitem {Colin80} {\sc Y. Colin de Veri\`{e}re}, {\em  Spectre conjoint d'op\'{e}rateurs
    pseudo-diff\'{e}rentiels qui commutent. II. Le cas int\'{e}grable},  Math. Z. {\bf 171} (1980), no. 1, p. 51-73.



\bibitem {lectureColin} {\sc Y. Colin de Veri\`{e}re}, {\em  M\'{e}thodes semi-classiques et th\'{e}orie
    spectrale}, Cours de DEA, 2006.


\bibitem {Vu-Ngoc99}{\sc S. V\~{u} Ng\d{o}c}, {\em  Quantum monodromy in integrable systems},
    Comm.  Math. Phys. {\bf 203} (1999), no. 2, p. 465-479.



\bibitem {Vu-Ngoc01} {\sc S. V\~{u} Ng\d{o}c}, {\em  Invariants symplectiques et semi-classiques des
    syst\`{e}mes int\'{e}grables  avec singularit\'{e}s}, S\'{e}minaire: \'{E}quations aux D\'{e}riv\'{e}es Partielles, 2000–2001, S\'{e}min. \'{E}qu. D\'{e}riv. Partielles, \'{E}cole Polytech.,
    Palaiseau, 2001, p. Exp. No. XII, 16.




\bibitem {Vu-Ngoc.S01} {\sc S. V\~{u} Ng\d{o}c}, {\em  Quantum monodromy and Bohr-Sommerfeld
    rules},  Lett. Math. Phys. {\bf 55} (2001), no. 3, p. 205-217, Topological and geometrical methods (Dijon, 2000).





\bibitem {Vu-Ngoc06} {\sc S. V\~{u} Ng\d{o}c}, {\em Syst\'{e}mes int\'{e}grables semi-classiques: du local au
    global, Panoramas et Synth\`{e}ses}, [Panoramas and Syntheses], vol. 22, Soci\'{e}t\'{e} Math\'{e}matique de France, Paris, 2006.


\bibitem {Vu-Ngoc09}{\sc S. V\~{u} Ng\d{o}c}, {\em  Quantum Birkhoff normal forms and semiclassical
    analysis}, Noncommutativity and singularities, Adv. Stud. Pure Math., vol. 55, Math. Soc. Japan, Tokyo, 2009, p. 99-
    116.




\bibitem {Weyl50} {\sc H. Weyl}, {\em  The theory of groups and quantum mechanics}, Library of
    Theoretical Physics, Dover, 1950, Translated from the german edition.




\end{thebibliography}
\end{document}